\documentclass[11pt]{article} 
\pdfoutput=1
\usepackage[pagebackref,letterpaper=true,colorlinks=true,pdfpagemode=none,urlcolor=blue,linkcolor=blue,citecolor=red,pdfstartview=FitH]{hyperref}

\usepackage
{
        amssymb,
        amsthm,
        amsmath,
        graphicx,
        fullpage,
        enumerate,
        times,
        caption,
        subcaption
}

\usepackage{amsfonts,latexsym,xspace,makeidx,bm}
\usepackage[usenames]{color}

\usepackage{boxedminipage}

\newcommand{\draft}{0}

\newcommand{\full}{1}

\newcommand {\calA}   {{\cal{A}}}
\newcommand {\calB}   {{\cal{B}}}
\newcommand {\calC}   {{\cal{C}}}
\newcommand {\calD}   {{\cal{D}}}

\newcommand {\calI}   {{\cal{I}}}
\newcommand {\calL}   {{\cal{L}}}

\newcommand {\calG}   {{\cal{G}}}

\newcommand {\calX}   {{\cal{X}}}

\newcommand {\NP}   {{\mathsf{NP}}}
\newcommand {\PTIME}   {{\mathsf{P}}}

\newcommand {\tensor}   {\otimes}

\newtheorem{theorem}{Theorem}[section]
\newtheorem{lemma}[theorem]{Lemma}
\newtheorem{proposition}[theorem]{Proposition}
\newtheorem{claim}[theorem]{Claim}
\newtheorem{corollary}[theorem]{Corollary}

\newtheorem{fact}[theorem]{Fact}

\newtheorem{remark}{Remark}[section]
\newtheorem{notation}{Notation}[section]

\newtheorem{definition}{Definition}

\definecolor{DSgray}{cmyk}{0,1,0,0}

\ifnum\draft=1
\newcommand{\rnote}[1]{\footnote{\color{magenta}Ryan: {#1}}}
\newcommand{\cnote}[1]{\footnote{\color{red}Chenggang: {#1}}}
\newcommand{\ynote}[1]{\footnote{\color{blue}Yuan: {#1}}}
\newcommand{\jnote}[1]{\footnote{\color{green}John: {#1}}}
\fi
\ifnum\draft=0
\newcommand{\rnote}[1]{{}}
\newcommand{\cnote}[1]{{}}
\newcommand{\ynote}[1]{{}}
\newcommand{\jnote}[1]{{}}

\newcommand{\R}{\mathbb R}

\newcommand{\Z}{\mathbb Z}

\usepackage{prettyref}

\newcommand{\savehyperref}[2]{\texorpdfstring{\hyperref[#1]{#2}}{#2}}

\newrefformat{eq}{\savehyperref{#1}{\textup{(\ref*{#1})}}}
\newrefformat{lem}{\savehyperref{#1}{Lemma~\ref*{#1}}}
\newrefformat{def}{\savehyperref{#1}{Definition~\ref*{#1}}}
\newrefformat{thm}{\savehyperref{#1}{Theorem~\ref*{#1}}}
\newrefformat{cor}{\savehyperref{#1}{Corollary~\ref*{#1}}}
\newrefformat{corr}{\savehyperref{#1}{Corollary~\ref*{#1}}}
\newrefformat{cha}{\savehyperref{#1}{Chapter~\ref*{#1}}}
\newrefformat{sec}{\savehyperref{#1}{Section~\ref*{#1}}}
\newrefformat{app}{\savehyperref{#1}{Appendix~\ref*{#1}}}
\newrefformat{tab}{\savehyperref{#1}{Table~\ref*{#1}}}
\newrefformat{fig}{\savehyperref{#1}{Figure~\ref*{#1}}}
\newrefformat{hyp}{\savehyperref{#1}{Hypothesis~\ref*{#1}}}
\newrefformat{alg}{\savehyperref{#1}{Algorithm~\ref*{#1}}}
\newrefformat{item}{\savehyperref{#1}{Item~\ref*{#1}}}
\newrefformat{step}{\savehyperref{#1}{step~\ref*{#1}}}
\newrefformat{conj}{\savehyperref{#1}{Conjecture~\ref*{#1}}}
\newrefformat{fact}{\savehyperref{#1}{Fact~\ref*{#1}}}
\newrefformat{prop}{\savehyperref{#1}{Proposition~\ref*{#1}}}
\newrefformat{claim}{\savehyperref{#1}{Claim~\ref*{#1}}}

\let\pref=\prettyref

\newcommand{\Sref}[1]{\hyperref[#1]{Section \ref*{#1}}}

\newcommand{\E}{\mathop{\bf E\/}}

\newcommand{\eps}{\epsilon}
\renewcommand{\phi}{\varphi}

\newcommand{\GI}{{\sf GI}}
\newcommand{\aut}{{\sf AUT}}
\newcommand{\val}{{\sf val}}

\newcommand{\floor}[1]{\left\lfloor #1 \right\rfloor}

\newcommand{\ignore}[1]{{}}

\newcommand{\id}{{\rm id}}
\newcommand{\diff} {\triangle}
\newcommand{\threexor}{{\sc 3XOR}\xspace}
\newcommand{\giso}{{\sc Giso}\xspace}
\newcommand{\robustgiso}{{\sc RobustGiso}\xspace}
\newcommand{\gisolong}{{\sc GraphIsomorphism}\xspace}
\newcommand{\factoring}{{\sc Factoring}\xspace}
\newcommand{\RXOR}{{\sc R3XOR}\xspace}

\newcommand{\poly}{\mathrm{poly}}
\newcommand{\polylog}{\mathrm{polylog}}
\newcommand{\co}{\colon}
\newcommand{\homog}[1]{\underline{#1}}

\newcommand{\pE}{\widetilde{\E}}
\newcommand{\littlesum}{\mathop{{\textstyle \sum}}}

\begin{document}

\title{Hardness of Robust Graph Isomorphism, \\
Lasserre Gaps, and Asymmetry of Random Graphs}
\author{
Ryan O'Donnell\thanks{Supported by NSF grants CCF-0747250 and CCF-1116594 and a grant from the MSR--CMU Center for Computational Thinking.}\\
Department of Computer Science\\
Carnegie Mellon University\\
\tt{odonnell@cs.cmu.edu}
\and
John Wright$^*$\thanks{Some of this research done with visiting Toyota Technological Institute at Chicago.}\\
Department of Computer Science\\
Carnegie Mellon University\\
\tt{jswright@cs.cmu.edu}
\and
Chenggang Wu\thanks{This work was supported in part by the National Basic Research Program of China Grant 2011CBA00300, 2011CBA00301, the National Natural Science Foundation of China Grant 61033001, 61061130540. Some of this research done while visiting Carnegie Mellon University.}\\
IIIS\\
Tsinghua University\\
\tt{wcg06@mails.tsinghua.edu.cn}
\and
Yuan Zhou$^{*\dagger}$\thanks{Also supported by a grant from the Simons Foundation (Award Number 252545).}\\
Department of Computer Science\\
Carnegie Mellon University\\
\tt{yuanzhou@cs.cmu.edu}
}

\maketitle

\abstract{Building on work of Cai, F\"urer, and Immerman~\cite{CFI92}, we show two hardness results for the Graph Isomorphism problem.   First, we show that there are pairs of nonisomorphic $n$-vertex graphs~$G$ and~$H$ such that any sum-of-squares (SOS) proof of nonisomorphism requires degree~$\Omega(n)$. In other words, we show an $\Omega(n)$-round integrality gap for the Lasserre SDP relaxation.  In fact, we show this for pairs $G$ and~$H$ which are not even $(1-10^{-14})$-isomorphic. (Here we say that two $n$-vertex, $m$-edge graphs $G$ and $H$ are $\alpha$-isomorphic if there is a bijection between their vertices which preserves at least $\alpha m$ edges.)  Our second result is that under the \RXOR~Hypothesis~\cite{Fei02} (and also any of a class of hypotheses which generalize the \RXOR~Hypothesis), the \emph{robust} Graph Isomorphism problem is hard. I.e.\ for every $\eps > 0$, there is no efficient algorithm which can distinguish graph pairs which are $(1-\eps)$-isomorphic from pairs which are not even  $(1-\eps_0)$-isomorphic for some universal constant $\eps_0$. Along the way we prove a robust asymmetry result for random graphs and hypergraphs which may be of independent interest.
}

\setcounter{page}{0}
\thispagestyle{empty}
\newpage

\section{Introduction}

The \gisolong problem is one of the most intriguing and notorious problems in computational complexity theory (we will also refer to it as \giso for short); we refer to~\cite{KST93,Bab95,AK05,Kob06,Cod11} for surveys. It was famously referred to as a ``disease'' over 35 years ago~\cite{RC77} and maintains its infectious status to this day. Together with \factoring,\rnote{I ignore the problems related to GISO like GA, and the problems related to Factoring like Discrete Log} it is one of the very rare problems\rnote{Yes, I realize Factoring is not exactly a language, but c'mon, the reader knows what we mean} in~$\NP$ which is not known to be in~$\PTIME$ but which is believed to not be $\NP$-hard~\cite{Bab85,BHZ87,Sch88} (according to standard complexity-theoretic assumptions).  Both problems also admit an algorithm with running time ``subexponential'' (or ``moderately exponential'') in the natural witness size.  In the case of \gisolong on $n$-vertex graphs, the natural witness size is $\log_2(n!) = \Theta(n \log n)$, but the best known algorithm due to Luks solves the problem in time $2^{O(\sqrt{n \log n})}$~\cite{BL83}.

In the same breath we might mention the problems {\sc GapSVP}$_{\sqrt{n}}$ (approximating the shortest vector in an $n$-dimensional lattice to factor~$\sqrt{n}$) and {\sc UniqueGames}$_\eps$ (the Unique Games problem proposed by Khot~\cite{Kho02}).  The former is not $\NP$-hard subject to standard complexity-theoretic assumptions~\cite{GG00,AR05}, though we don't know any subexponential-time algorithm.\rnote{I don't think?}  The latter has a subexponential-time algorithm~\cite{ABS10}; whether it is $\NP$-hard or in $\PTIME$ (or neither) is hotly contested.  The potential hardness of \factoring and {\sc GapSVP} --- even under certain average-case distributions --- is well enough entrenched that many cryptographic protocols are based on it.\rnote{do we need citations?  probably not.}  (The same is true of random \threexor with noise, more on which later.) On the other hand, for \gisolong and {\sc UniqueGames}$_\eps$ we do not know any way of generating ``hard-seeming instances''; indeed, some experts have speculated that \gisolong may be in~$\PTIME$, or at least have a $2^{\polylog(n)}$-time algorithm.\rnote{cite?}

In this paper we investigate hardness results for the \giso problem.\rnote{Yes, I switched without comment to the short form here.}  Since \giso may well be in~$\PTIME$, let us discuss what this may mean.  One direction would be to show that \giso is hard for small complexity classes. This has been pursued most successfully by Tor\'{a}n~\cite{Tor04}, who has shown that \giso is hard for the class~$\mathsf{DET}$.  This is essentially the class of problems equivalent to computing the determinant; it contains~$\#\mathsf{L}$ and is contained in~$\mathsf{TC}^1$. It is not known whether \giso is $\PTIME$-hard.

In this paper, however, we are concerned with hardness results well above~$\PTIME$.  We have two related main results. The first is that solving \giso via the Lasserre/SOS hierarchy requires $2^{\Omega(n)}$ time (i.e., $\Omega(n)$ rounds/degree).  This generalizes the result of Cai, F\"urer, and Immerman~\cite{CFI92} showing that the frequently effective $o(n)$-dimensional Weisfeiler--Lehman algorithm fails to solve~\giso; it also gives even more evidence that any subexponential-time algorithm for \giso requires algebraic, non-local techniques.

Our second result is concerned with the problem of \emph{robust} graph isomorphism, \robustgiso.  Roughly speaking, this is the following problem: given two graphs which are almost isomorphic, find an ``almost-isomorphism''.  \robustgiso is strictly harder than \giso and the fact that it concerns ``isomorphisms with noise'' seems to rule out all algebraic techniques.  We show that \robustgiso is at least as hard as random \threexor with noise; hence \robustgiso has no polynomial-time algorithm assuming the well-known \RXOR Hypothesis of Feige~\cite{Fei02}.  In fact, it's possible that the \RXOR problem requires $2^{n^{1-o(1)}}$ time\rnote{I don't think you can put $2^{\Omega(n)}$ here in light of Blum--Kalai--Wasserman}, which would mean that \robustgiso is much harder than \giso itself.

\subsection{SOS/Lasserre gaps}

The most well-known heuristic for \gisolong (and the basis of most practical algorithms --- e.g., ``nauty''~\cite{McK81}) is the Weisfeiler--Lehman~(WL) algorithm~\cite{Wei76} and its ``higher dimensional'' generalizations.  To describe the basic algorithm\rnote{please check I am actually describing it correctly...} we need the notion of a colored graph. This is simply a graph, together with a function mapping the vertices to a finite set of colors; equivalently, a graph with its vertices partitioned into ``color classes''.   Isomorphisms involving colored graphs are always assumed to preserve colors.  Let~$G$ be a colored graph on the $n$-vertex set~$V$.  A \emph{color refinement step} refers to the following procedure: for each $v \in V$, one determines the multiset $C_v$ of colors in the neighborhood of~$v$;\rnote{including~$v$; this isn't a disquisition on WL so I am trying to be brief} then one recolors each~$v$ with color~$C_v$.\rnote{that's pretty terse, I know}  Now the basic WL algorithm, when given graphs~$G$ and~$H$, repeatedly applies refinement to each of them until the colorings stabilize.  (Initially, the graphs are treated as having just one color class.)  At the end, if $G$ and $H$ have the same number of vertices of each color the WL algorithm outputs that they are ``maybe isomorphic''; otherwise, it (correctly) outputs that they are ``not isomorphic''.

Note that after the initial refinement step, a graph's vertices are colored according to their degree.  Thus two $d$-regular graphs are always reported as ``maybe isomorphic'' by the basic WL algorithm.  On the other hand, the heuristic is powerful enough to work correctly for all trees\rnote{folklore?} and for almost all $n$-vertex graphs in the Erd\H{o}s--R\'enyi $\calG(n,1/2)$ model~\cite{BES80,BK79}.\rnote{Is this the notation we use?}  (We say the heuristic ``works correctly'' on a graph~$G$ if the stabilized coloring for~$G$ is distinct from the stabilized coloring of any graph not isomorphic to~$G$.)  To overcome WL's failure for regular graphs, several researchers independently introduced the ``$k$-dimensional generalization'' WL$^k$ (see~\cite{Wei76,CFI92} for discussion).  Briefly, in the WL$^k$ heuristic, each $k$-tuple of vertices has a color, and color refinement involves looking at all ``neighbors'' of each $k$-tuple of vertices~$(v_1, \dots, v_k)$ (where the neighbors are all tuples of the form $(v_1, \dots, v_{i-1}, u_i, v_{i+1}, \dots, v_k)$, where $\{u_i,v_i\}$ is an edge).\rnote{should we bother including this parenthesis?}  The WL$^k$ heuristic can be performed in time $n^{k+O(1)}$ and is thus a polynomial-time algorithm for any constant~$k$.

The WL$^k$ heuristic is very powerful.  For example, it is known to work correctly in polynomial time for all graphs which exclude a fixed minor~\cite{Gro12}, a class which includes all graphs of bounded treewidth or bounded genus.\rnote{Also works for colored graph isomorphism with color classes of size at most 3.}  Spielman's $2^{\widetilde{O}(n^{1/3})}$-time graph isomorphism algorithm~\cite{Spi96a} for strongly regular graphs is achieved by WL$^k$ with $k = \widetilde{O}(n^{1/3})$.  The WL$^k$ algorithm with $k = O(\sqrt{n})$\rnote{no log factors according to CFI} is also a key component in the $2^{O(\sqrt{n \log n})}$-time \giso algorithm~\cite{BL83}.  Throughout the '80s there was some speculation that \giso might be solvable on all graphs by running the WL$^k$ algorithm with $k = O(\log n)$ or even $k = O(1)$.  However this was disproved in the notable work of Cai, F\"urer, and Immerman~\cite{CFI92}, which showed the existence of nonisomorphic $n$-vertex\rnote{$3$-regular} graphs~$G$ and~$H$ which are not distinguished by WL$^k$ unless $k = \Omega(n)$.\footnote{Actually, $G$ and $H$ are colored graphs in~\cite{CFI92}'s construction, with each color class having size at most~$4$.  It is often stated that the colors can be replaced by gadgets while keeping the number of vertices~$O(n)$. We do not find this to be immediately obvious. However it does follow from the asymmetry of random graphs, as we will see in this paper.}\rnote{Right?  Or is it obvious?}

The \gisolong problem can be thought of as kind of constraint satisfaction problem (CSP), and readers familiar with LP/SDP hierarchies for CSPs might see an analogy between $k$-dimensional~WL and level-$k$ LP/SDP relaxations.  A very interesting recent work of Atserias and Maneva~\cite{AM13} (see also~\cite{GO12}) shows that this is more than just an analogy --- it shows that the power of WL$^k$ is precisely sandwiched between that of the $k$th and $(k+1)$st level of the canonical \emph{Sherali--Adams LP hierarchy}~\cite{SA90}.  (In fact, it had long been known~\cite{RSU94} that WL$^1$ is equivalent in power to the basic LP relaxation of \giso.)  This gives a very satisfactory connection between standard techniques in optimization algorithms and the best known non-algebraic/local techniques for \giso.

This connection raises the question of whether stronger LP/SDP hierarchies might prove more powerful than WL$^k$ in the context of \giso.  The strongest such hierarchy known\rnote{Well, I'm not sure of the status of Bienstock--Zuckerman} is the ``SOS (sum-of-squares) hierarchy'' due to Lasserre~\cite{Las00} and Parrilo~\cite{Par00}. Very recent work~\cite{BBH+12,OZ13,KOTZ13} in the field of CSP approximability has shown that $O(1)$~levels of the SOS hierarchy can succeed where $\omega(1)$ levels of weaker SDP hierarchies fail; in particular, this holds for the hardest known instances of {\sc UniqueGames}$_\eps$~\cite{BBH+12}.  This raises the question of whether there might be a subexponential-time algorithm based SOS which solves \gisolong.

We answer this question negatively.  Our first main result is that a variant of the Cai--F\"urer--Immerman instances also fools $\Omega(n)$ levels of the SOS hierarchy.  In fact, we achieve a ``constant factor Lasserre gap with perfect completeness''.  To explain this we require a definition:
\begin{definition}
    Let $G$ and $H$ be nonempty $n$-vertex graphs.\rnote{should we also define here for hypergraphs?  save it for later?  not define it at all?}  For $0 \leq \beta \leq 1$, we say that a bijection $\pi \co V(G) \to V(H)$ is an \emph{$\alpha$-isomorphism} if
\ifnum\full=1
\[
\fi
\ifnum\full=0
$
\fi
        \frac{|\{(u,v) \in E(G) : (\pi(u), \pi(v)) \in E(H)\}|}{\max\{|E(G)|, |E(H)|\}} \geq \alpha.
\ifnum\full=1
\]
\fi
\ifnum\full=0
$
\fi
    In this case we say that $G$ and $H$ are \emph{$\alpha$-isomorphic}.
\end{definition}
Observe that this definition is symmetric in $G$ and $H$.  The two graphs are isomorphic if and only if they are $1$-isomorphic.  We will almost always\rnote{always?} consider the case that $G$ and $H$ have the same number of edges. We prove:
\begin{theorem}                                     \label{thm:lasserre-gap}
    For infinitely many $n$, there exist pairs of $n$-vertex, $O(n)$-edge graphs~$G$ and~$H$ such that:
    \begin{itemize}
        \item $G$ and $H$ are not $(1-10^{-14})$-isomorphic;
        \item any SOS refutation of the statement ``$G$ and $H$ are isomorphic''\footnote{When naturally encoded.} requires degree~$\Omega(n)$.
    \end{itemize}
\end{theorem}

A word on our techniques.  The essence of the Cai--F\"urer--Immerman construction is to take a $3$-regular expander graph and replace each vertex by a certain $10$-vertex gadget (originally appearing in~\cite{Imm81} and also sometimes called a ``F\"urer gadget'').  This gadget is closely related to $3$-variable equations modulo~$2$ (as observed by several authors, e.g.~\cite{Tor04}); indeed, it may be described as the ``label-extended graph'' of the \threexor constraint.  The reader may therefore recognize the~\cite{CFI92} WL$^{k}$ lower bound as stemming from the difficulty of refuting unsatisfiable, expanding \threexor CSP instances by ``local'' means.  This should make our \pref{thm:lasserre-gap} look plausible in light of the Grigoriev~\cite{Gri01} and Schoenebeck~\cite{Sch08} SOS/Lasserre lower bounds.

Nevertheless, obtaining \pref{thm:lasserre-gap} is not automatic.  For one, we still lack a complete theory of reductions within the SOS hierarchy (though see~\cite{Tul09}). Second, the pair of graphs constructed by~\cite{CFI92} only differ by one edge. More tricky is the issue of removing the ``colors'' from the~\cite{CFI92} construction.  We do not see an easy gadget-based way of doing this without sacrificing on the $\Omega(n)$~degree.  To handle this we have to: a)~modify the~\cite{CFI92} construction somewhat to make the two graphs differ by a constant fraction; b)~prove that random (hyper)graphs are ``robustly asymmetric'' --- i.e., ``far from having nontrivial automorphisms''; c)~use the robust asymmetry property to remove the ``colors''.  The result in b), described below in \pref{sec:asymm-overview}, qualitatively generalizes work of Erd\H{o}s and R\'enyi \cite{ER63}
and may be of independent interest.

In an independent work, Codenotti, Schoenebeck, and Snook~\cite{CSS14} have shown a conclusion similar to our Theorem~\ref{thm:lasserre-gap}.  Their main result is that there are expander graphs $G$ and $H$ which are not isomorphic, but any SOS refutation of the statement ``$G$ and $H$ are isomorphic'' requires degree~$\Omega(n)$.  As in our work, their proof combines the Cai--F\"urer--Immerman construction with the Schoenebeck~\cite{Sch08} SOS/Lasserre lower bounds.

\subsection{Robust graph isomorphism}
Our second main result concerns the \robustgiso problem.  To introduce it, we might imagine an algorithm trying to recover an isomorphism between $G$ and $H$, where $H$ is formed by permuting the vertices of~$G$ and then introducing a small amount of \emph{noise} --- say, adding and deleting an~$\eps$ fraction of edges.  Thinking of \giso as a CSP, we are concerned with finding ``almost-satisfying'' solutions on ``almost-satisfiable'' instances.  For example, suppose we are given graphs~$G$ and~$H$ which are promised to be $(1-\eps)$-isomorphic.  Can we efficiently find a $(1-2\eps)$-isomorphism?  A $(1-\sqrt{\eps})$-isomorphism?  A $(1-\frac{1}{\log(1/\eps)})$-isomorphism?
\begin{definition}
    We say an algorithm $\calA$ solves the \robustgiso problem if there is a function $r \co [0,1] \to [0,1]$ satisfying $r(\eps) \to  0$ as $\eps \to 0^+$ such that whenever $\calA$ is given any pair of graphs which are $(1-\eps)$-isomorphic,~$\calA$ outputs a $(1-r(\eps))$-isomorphism between them.\footnote{We could also consider the easier task of distinguishing pairs which are $(1-\eps)$-isomorphic from pairs which are not $(1-r(\eps))$-isomorphic. In fact, our hardness result will hold even for this easier problem.}
\end{definition}
\begin{remark}
    In particular, $\calA$ must solve the \giso problem, because given isomorphic graphs with at most~$m$ edges we can always take $\eps > 0$ small enough so that $r(\eps) < 1/m$.\rnote{in particular, I don't think that formally speaking we need to require $r(0) = 0$}
\end{remark}

The analogous problem of robust satisfaction algorithms for CSPs over constant-size domains was introduced by Zwick~\cite{Zwi98a} and has proved to be very interesting.  Guruswami and Zhou~\cite{GZ11} conjectured that the CSPs which have efficient robust algorithms (subject to $\PTIME \neq \NP$) are precisely those of ``bounded width'' --- roughly speaking, those that do not encode equations over abelian groups.  This conjecture was recently confirmed by Barto and Kozik~\cite{BK12a}, following partial progress in~\cite{KOT+12,DK12a}.  The problem of \robustgiso was first introduced explicitly by Wu et al.~\cite{WYZV13}; their results included, among another things, a nontrivial efficient \robustgiso algorithm for trees.  It is an interesting question to find efficient \robustgiso algorithms for more classes of graphs on which \giso is solvable --- e.g., planar graphs, or bounded treewidth graphs.

The graph isomorphism seems to share some of the flavor of ``unbounded width'' CSPs such as \threexor; these CSPs have the property that special algebraic methods (namely, Gaussian elimination) are available on satisfiable instances, but these methods break down once there is a small amount of noise.  Indeed, the $2^{O(\sqrt{n \log n})}$-time algorithm for \giso is a somewhat peculiar mix of group theory and ``local methods'' (namely, Weisfeiler--Lehman).  Generalizing from \giso to \robustgiso seems like it might rule out applicability of group-theoretic methods, thereby making the problem much harder.  Our second main theorem in a sense confirms this.\rnote{Mention that this sort of resolves an open problem from \cite{WYZV13}??} Roughly speaking, it shows that \robustgiso is hard assuming it is hard to distinguish random \threexor instances from random instances with a planted solution and slight noise:
\begin{theorem}                                     \label{thm:r3xor-hardness}
    Assume Feige's \RXOR Hypothesis~\cite{Fei02}.  Then there is no polynomial-time algorithm for \robustgiso.  More precisely, there exists $\epsilon_0 > 0$, such that suppose there exists $\epsilon > 0$ and a $t(n)$-time algorithm which can distinguish $(1-\eps)$-isomorphic $n$-vertex, $m$-edge graph pairs from pairs which are not even $(1-\epsilon_0)$-isomorphic (where $m = O(n)$).  Then there is a  universal constant $\Delta \in \Z^+$ and a $t(O(n))$-time algorithm which outputs ``typical'' for almost all $n$-variable, $\Delta n$-constraint instances of the \threexor problem, yet which never outputs ``typical'' on instances which are $(1-\Theta(\eps))$-satisfiable.
\end{theorem}
Here we refer to:
\paragraph{Feige's \RXOR Hypothesis.} \emph{For every fixed $\eps > 0$, $\Delta \in Z^+$, there is no polynomial time algorithm which on almost all \threexor instances with $n$ variables and $m = \Delta n$ constraints outputs ``typical'', but which never outputs ``typical'' on instances which an assignment satisfying at least $(1-\eps)m$ constraints.}
\ignore{
\paragraph{Feige's \RXOR Hypothesis.} \emph{For every fixed $\eps > 0$, for $\Delta \in Z^+$ a sufficiently large constant, there is no polynomial time algorithm which on almost all \threexor instances with $n$ variables and $m = \Delta n$ constraints outputs ``typical'', but which never outputs ``typical'' on instances with an assignment satisfying at least $(1-\eps)m$ constraints.}}

\begin{remark}
    The reader may think of the output ``typical'' as a \emph{certification} that the \threexor instance has no $(1-\eps)$-satisfying solution.  Note that with high probability the random \threexor instance will not even have a $.51$-satisfying solution.  Feige originally stated his hypothesis for the random 3SAT problem rather than the random 3XOR problem, but he showed the conjectures are equivalent. See also the work of Alekhnovich~\cite{Ale03}.
\end{remark}

Feige's \RXOR Hypothesis is a fairly well-believed conjecture.  The variation in which the XOR constraints may involve any number of variables (not just~$3$) is called {\sc LPN} (Learning Parities with Noise) and is believed to be hard even with any $m = \poly(n)$ constraints.  The further variation which has linear equations modulo a large prime rather than modulo~$2$ is called {\sc LWE} (Learning With Errors) and forms the basis for a large body of cryptography. (See~\cite{Reg05} for more on LPN and LWE.) The fastest known algorithm for solving Feige's \RXOR problem seems to be the $2^{O(n/\log n)}$-time algorithm of Blum, Kalai, and Wasserman~\cite{BKW03}.  Thus it's plausible that \robustgiso requires $2^{n^{1-o(1)}}$ time, which would make it a much more difficulty problem than \giso.

In fact, we are able to base \pref{thm:r3xor-hardness} on a larger class of hypotheses which generalize Feige's \RXOR Hypothesis.  This is inspired by the recent work of Chan~\cite{Cha13}, who showed approximation resistance for all predicates supported on a pairwise independent abelian subgroup.  In light of this, we show that if no polynomial-time algorithm can distinguish random instances from nearly satisfiable instances of \emph{any} such predicate, then \pref{thm:r3xor-hardness} holds.  Proving \pref{thm:r3xor-hardness} in this more general setting requires some additional tricks.  For example, we give a novel gadget construction whose only automorphisms correspond to group elements.
\ifnum\full=1
See \pref{sec:feige} for more details.
\fi
\ifnum\full=0
See the full version of this paper for more details.
\fi

\ifnum\full=0
We also mention some related literature on approximate graph isomorphism in the full version of this paper, but it is omitted here due to space constraints.
\fi
\ifnum\full=1
We close by mentioning some related literature on approximate graph isomorphism.  The problem of finding a vertex permutation which maximizes the number of edge overlaps (or minimizes the number of edge/nonedge overlaps) was perhaps first discussed by Arora, Frieze, and Kaplan~\cite{AFK02}.  They gave an additive PTAS in the case of dense graphs ($m = \Omega(n^2)$).  Arvind et al.~\cite{AKKV12} recently defined and studied several variants of the approximate graph isomorphism problem.  Some of their results concern the case in which $G$ and $H$ have noticeably different numbers of edges and one isn't ``punished'' for uncovered edges in~$H$.  This kind of variant is more like approximate \emph{subgraph} isomorphism, and is much harder. (E.g., when $G$ is a $k$-clique and $H$ is a general graph the problem is roughly equivalent to the notorious {\sc Densest-}$k${\sc Subgraph} problem.)  The result of theirs which is most relevant to the present work involves hardness of finding approximate isomorphisms in \emph{colored} graphs. In particular, Arvind eta al.\ prove the following:
\begin{theorem}                                     \label{thm:arvind-et-al}
    (\cite{AKKV12}.)  There is a linear-time reduction from 2XOR (modulo~$2$) instances $\calI$ to pairs of \emph{colored} graphs $G, H$ such that $G$ and $H$ are $\alpha$-isomorphic if and only if $\calI$ has a solution satisfying at least an $\alpha$-fraction of constraints.  In particular it is $\NP$-hard to approximate $\alpha$-isomorphism for colored graphs to a factor exceeding $\frac{11}{12}$ and {\sc UniqueGames}-hard to approximate it to a factor exceeding~$.878$ (by results of~\cite{TSSW00,Has01}, \cite{KKMO07,MOO10} respectively).
\end{theorem}
In particular, the theorem holds for colored graphs in which each color class contains at most~$4$ vertices.  However, we do not see any way of eliminating the colors and getting the analogous inapproximability results for the usual \giso problem without using gadgets that would destroy the constant-factor gap.
\fi

\ifnum\full=0
\vspace{-2ex}
\fi
\subsection{Robust asymmetry of random graphs}   \label{sec:asymm-overview}

One of our main technical contributions is showing that random graphs are ``robustly asymmetric''. In doing so, we generalize the concept of an \emph{asymmetric graph}, which is a graph whose only automorphism is the trivial identity automorphism.  A line of research (see, e.g., \cite{ER63, Bol82, MW84, KSV02}) has shown that several distributions of random graphs produce asymmetric graphs with high probability.  In their well-known $\calG(n,p)$ model, Erd\H{o}s  and R\'enyi~\cite{ER63} proved that for $\frac{\ln n}{n} \leq p \leq 1- \frac{\ln n}{n}$, $\calG(n, p)$ is asymmetric with high probability.  If we instead consider a uniformly random $d$-regular $n$-vertex graph, the sequence of works~\cite{Bol82,MW84,KSV02} shows that we get an asymmetric graph with high probability for any $3 \leq d \leq n - 4$. In this work we will work with a third variant, the $\calG_{n, m}$ model, in which a graph is chosen uniformly at random from all simple graphs with~$n$ vertices and~$m$ edges.

Given a graph $G$ and a permutation $\pi$ over $V(G)$, \rnote{the following clashes with notation we have already defined...} we call $\pi$ an $\alpha$-automorphism if the application of $\pi$ on $G$ preserves at least an $\alpha$ fraction of the edges. A graph $G$ is $(\beta, \gamma)$-asymmetric if any $\gamma$-automorphism $\pi$ has more than a fraction of $(1-\beta)$ fixed points. Intuitively, when $\beta = 1/n$, $\gamma = 0$, the property is exactly the asymmetry property; when $\beta$ and $\gamma$ become larger, the property requires that any permutation that is far from identity is far from an automorphism for the graph. We encourage the reader to refer to \pref{sec:prelim} for the precise definitions.

In this paper, we show the following robust asymmetric property of $\calG_{n, m}$.
\begin{theorem}\label{thm:robust-asymmetry}
For large enough $n$, suppose that $m=cn$, where $10^4\leq c\leq n/10^{10}$. Let $\beta^*=\max\{e^{-c/6},\frac{1}{n}\}$.  With probability $(1 - n^{-15})$, for all $\beta$ such that $\beta^* \leq \beta \leq 1$, $\calG_{n,m}$ is $(\beta,\beta/240)$-asymmetric. 
\end{theorem}
A couple of comments are in order. First, an $\exp(-O(c))$ lower bound on $\beta$ is necessary. This is because there are at least $\lfloor \exp(-O(c)) \cdot n \rfloor$ isolated vertices in $\calG_{n, m}$ with high probability. The permutations which only permute these isolated vertices are $1$-automorphisms. Therefore, with high probability, $\calG_{n, m}$ is not $(\exp(-\omega(c)), 0)$-asymmetric. Second, it is possible to extend our theorem to the $\calG(n, p)$ model  by showing that there exists a constant $C > 0$, such that for $\frac{C}{n} < p < \frac{1}{C}$, with high probability, $\calG(n, p)$ is $(\beta, \beta / 240)$-asymmetric for all $\beta \geq \max\{\exp(-\frac{p n}{C}), \frac{1}{n}\}$. Third, when $c \geq 6 \ln n$ (or, when $p \geq \frac{C \ln n}{ n}$ in the $\calG(n, p)$ model), we can let $\beta = \frac{1}{n}$ and obtain that $\calG_{n, m}$ ($\calG(n, p)$, respectively) is asymmetric with high probability -- a result in the flavor of \cite{ER63}. Finally, we do not work hard to optimize the constants in the theorem statement; we believe a more careful analysis would bring them down to more civilized numbers, but it is still interesting to explore the limits of these constants.

Now we briefly explain our proof techniques. Let us consider the case where $c$ is a big constant and $\beta = 1$, so that we only need to worry about the permutations without fixed points. We would like to show that, for every such permutation $\pi$, $\Pr_{G\sim \calG_{n, m}}[\pi \text{~is a $\frac{1}{240}$-automorphism for $G$}] \ll \frac{1}{n!}$, and therefore we can union bound over all such possible permutations. In order to do this, from all ${n \choose 2}$ possible edges, we construct $\Omega(n^2)$ disjoint pairs of edges $(e, e')$, which we will refer to as ``bins'', such that $\pi(e) = e'$. We call a bin ``half-full'' if exactly one edge in the pair is selected in G. It is easy to see that whenever there are more than $\frac{m}{120}$ half-full bins, $\pi$ cannot be a $\frac{1}{240}$-automorphism. At this point, we would like for $\Pr_G \left[\#\text{half-full bins}<\frac{m}{120}\right] \ll \frac{1}{n!}$, and this is easy to show. Unfortunately, this method does not work when $\beta = \frac{1}{2}$. To see why, let $\pi$ be a permutation with $\frac{n}{2}$ fixed points. The probability that every edge in $G$ has fixed points of $\pi$ for its endpoints is roughly $2^{-2m} = 2^{-O(n)}$. Therefore we have $\Pr_G[\pi\text{~is an automorphism for $G$}] \geq 2^{-O(n)}$, and this is not enough for the application of union bound (since there are more than $(n/2)! = 2^{\Omega(n\log n)}$ such permutations). A possible fix to this problem is: we first show that with high probability $(1 - n^{-\omega(1)})$, for every $\pi$ with $\frac{n}{2}$ fixed points, there are many edges of $G$ with at least one end point not fixed by $\pi$; then, conditioned on this event, we show the probability that a fixed $\pi$ is not a $\frac{1}{480}$-automorphism is small enough for the union bound method. The actual proof is more involved, and it is also technically challenging to work with $c$ as large as $\Omega(n)$, and $\beta$ as small as $\frac{1}{n}$.

Finally, for our application to the \gisolong problem, we need to extend \pref{thm:robust-asymmetry} to hypergraphs. More details on robust asymmetry of random hypergraphs can be found in \pref{sec:robust-asymmetry}.

\subsection{Organization of the paper}

In \pref{sec:prelim}, we introduce the notations and the SOS/Lasserre hierarchy. In \pref{sec:reduction-3xor}, we describe a reduction from \threexor to \giso. The completeness and soundness lemmas for reduction are proved in \pref{sec:completeness-main} and \pref{sec:soundness-main} respectively. In \pref{sec:robust-asymmetry}, we prove robust asymmetry property for random graphs and random hypergraphs. 
\ifnum\full=1
In \pref{sec:feige}, we generalize our \robustgiso hardness result to be based on a larger class of hypotheses. 
\fi
\ifnum\full=0
The generalization of our \robustgiso hardness result to be based on a larger class of hypotheses is deferred to the full version of this paper due to space constraints.
\fi
We discuss some of the future directions in \pref{sec:conclusions}.

\ifnum\full=0
\vspace{-2ex}
\fi
\paragraph{Proofs of the main theorems.} \pref{thm:lasserre-gap} follows from \pref{thm:sch08}, \pref{lem:sos-completeness} and \pref{lem:soundness}, by choosing $c = 10^6$. \pref{thm:r3xor-hardness} follows from \pref{lem:completeness} and \pref{lem:soundness}, by choosing $c = \max\{10^6, \Delta\}$.

\section{Preliminaries}\label{sec:prelim}
We will be working with undirected graphs and hypergraphs, both of which will be denoted by $G = (V, E)$.  Here, an undirected edge $e \in E$ is a set of $2$ vertices $\{i, j\}$ for graphs and a set of $k$ vertices $\{i_1, i_2, \ldots, i_k\}$ for $k$-uniform hypergraphs. When $G$ is an directed graph, we use $(i, j)$ to denote a directed edge. We also use the notation $V(G)$ to denote the vertex set of $G$, and $E(G)$ to denote the edge set of $G$. 

For any two undirected graphs (or hypergraphs) $G = (V(G), E(G))$ and $H = (V(H), E(H))$ with the same number of vertices, and for any bijection $\pi: V(G) \to V(H)$, let 
\ifnum\full=1
\[
\fi
\ifnum\full=0
$
\fi
\GI(G, H; \pi) = \frac{|\{e \in E(G) : \pi(e) \in E(H)\}|}{\max\{|E(G), E(H)|\}},
\ifnum\full=1
\]
\fi
\ifnum\full=0
$
\fi
where $\pi(e)$ is the edge obtained by applying $\pi$ on each vertex incident to $e$. Let
\ifnum\full=1
\[
\fi
\ifnum\full=0
$
\fi
\GI(G, H) = \max_{\pi: V(G) \to V(H)} \GI(G, H; \pi) .
\ifnum\full=1
\]
\fi
\ifnum\full=0
$
\fi
We say that an edge $e \in E(G)$ is \emph{satisfied} by $\pi$ if $\pi(e) \in E(H)$. We call $\pi$ an \emph{$\alpha$-isomorphism} for $G$ and $H$ if $\GI(G, H; \pi) \geq \alpha$, and we say $G$ and $H$ are \emph{$\alpha$-isomorphic} if $\GI(G, H) \geq \alpha$.
\ifnum\full=1

\fi
For any permutation $\pi: V(G) \to V(G)$, let 
\ifnum\full=1
\[
\fi
\ifnum\full=0
$
\fi
\aut(G; \pi) = \GI(G, G; \pi).
\ifnum\full=1
\]
\fi
\ifnum\full=0
$
\fi
We say that $\pi$ is an \emph{$\alpha$-automorphism} for $G$ if $\aut(G; \pi) \geq \alpha$. 

Given a permutation $\pi$ over the set $V$, an element $i \in V$ is a \emph{fixed point} of $\pi$ if $\pi(i) = i$.

\begin{definition}
A graph (possibly hypergraph) $G$ is \emph{$(\beta, \gamma)$-asymmetric} if, for any permutation $\pi$ on the vertex set of $G$ that has at most $(1 - \beta)$ fraction of the vertices as fixed points, we have $\aut(G; \pi) < 1 - \gamma$. 
\end{definition}

We extend the $\calG_{n, m}$ random graph model to hypergraphs as follows. Let $\calG^{(k)}_{n, m}$ be the uniform distribution over all ${{n \choose k} \choose m}$ simple $k$-uniform hypergraphs with $n$ vertices and $m$ edges.

A \threexor instance $\calC$ is a collection of equations $C_1, C_2, \dots, C_m$ over the variable set $\calX$. Each equation $C_i$ is of the form $x_{j_1} + x_{j_2} + x_{j_3} = b$ where $x_{j_1}, x_{j_2}, x_{j_3}$ are the variables from $\calX$, $b \in \Z_2$. Given an assignment $\tau : \calX \to \Z_2$, let $\val(\calC; \tau)$ be the fraction of equations in $\calC$ satisfied by $\tau$. Let $\val(\calC) = \max_{\tau: \calX \to \Z_2} \val(\calC; \tau)$. 

\ifnum\full=1
\subsection{SOS/Lasserre hierarchy}
\fi

\ifnum\full=0
\vspace{-2ex}
\paragraph{SOS/Lasserre hierarchy.}
\fi

\ifnum\full=1

One way to formulate the SOS/Lasserre hierarchy is via the pseudo-expectation view. We briefly recall the formulation as follows. More discussion about this view can be found in \cite{BBH+12}. 

We consider the feasibility of a system over $n$ variables $(x_1, x_2, \dots, x_n) \in \R^n$ with the following constraints: $P_i(x) = 0$ for $i = 1, 2, \dots, m_{{P}}$ and $Q_i(x) \geq 0$ for $j = 1, 2, \dots, m_{{Q}}$, where all the $P_i, Q_i$ polynomials are of degree at most $d$. For $r \geq d$, the degree-$r$ SOS/Lasserre hierarchy finds a pseudo-expectation operator $\pE[\cdot]$ defined on the space of real polynomials of degree at most $r$ over intermediates $x_1, x_2, \dots, x_n$ such that 
\begin{itemize}
\item $\pE[1] = 1$;
\item $\pE[\alpha p + \beta q] = \alpha \pE[p] + \beta \pE[q]$ for all real numbers $\alpha, \beta$ and all polynomials $p$ and $q$ of degree at most $r$;
\item $\pE[p^2] \geq 0$ for all polynomials $p$ of degree at most $r/2$;
\item $\pE[P_i \cdot q] = 0$ for all $i = 1, 2, \dots, m_P$ and all polynomials $q$ such that $P_i \cdot q$ is of degree at most $r$;
\item $\pE[Q_i \cdot p^2] \geq 0$ for all $i = 1, 2, \dots, m_Q$ and all polynomials $p$ such that $Q_i \cdot p^2$ is of degree at most $r$.
\end{itemize}
We call any operator $\pE[\cdot]$ a normalized linear operator if it has the first two properties listed above.

\paragraph{SOS/Lasserre hierarchy for \threexor.} Let $\calC$ be a \threexor instance on variable set $\calX$. The degree-$r$ SOS/Lasserre hierarchy for the natural integer programming for (the satisfiability of) $\calC$ is to find a normalized linear pseudo-expectation operator $\pE[\cdot]$ defined on the space of polynomials of degree at most $r$ over the indeterminates $(A[x \mapsto a])_{x \in \calX, a \in \Z_2}$ associated to~$\calC$, such that
    \begin{enumerate}
        \item$\pE[(A[x \mapsto a]^2 - A[x \mapsto a]) \cdot q] = 0$ for all $x \in \calX$, $a \in \Z_2$, and polynomials~$q$;
        \item $\pE[(A[x \mapsto 0] + A[x \mapsto 1] - 1) \cdot q] = 0$ for all $x \in \calX$ and polynomials~$q$;
        \item$\pE[(\littlesum_{\alpha_C \text{\ satisfying $C$}} A[x_1 \mapsto \alpha_C(x_1)]A[x_2 \mapsto \alpha_C(x_2)]A[x_3 \mapsto \alpha_C(x_3)] - 1) \cdot q] = 0$ for each $C \in \calC$ involving variables $x_1, x_2, x_3$ and all polynomials~$q$;
        \item$\pE[p^2] \geq 0$ for all polynomials~$p$.
    \end{enumerate}
We say there is degree-$r$ SOS refutation for the satisfiability of $\calC$ if the pseudo-expectation operator with properties listed above does not exist.

\paragraph{SOS/Lasserre hierarchy for \giso.}  Let $G = (V(G), E(G))$ and $H = (V(H), E(H))$ be two graphs such that $|V(G)| = |V(E)|$, $|E(G)| = |E(H)|$. The degree-$r$ SOS/Lasserre hierarchy for the natural integer programming the isomorphism problem between $G$ and $H$ is to find a normalized linear pseudo-expectation operator $\pE[\cdot]$ on the space of real polynomials of degree at most~$r$ over the indeterminates $(\Pi[u \mapsto v])_{u \in V(G), v \in V(H)}$ such that:
    \begin{enumerate}[a.]
        \item $\pE[(\Pi[u \mapsto v]^2 - \Pi[u \mapsto v]) \cdot q] = 0$ for all $u \in V(G)$, $v \in V(H)$, and polynomials~$q$;
        \item $\pE[(\sum_{v \in V(H)} \Pi[u \mapsto v] - 1) \cdot q] = 0$ for all $u \in V(G)$ and polynomials~$q$; and similarly, $\pE[(\sum_{u \in V(G)} \Pi[u \mapsto v] - 1) \cdot q] = 0$ for all $v \in V(H)$ and polynomials~$q$;
        \item \label{item:SOS-GISO} $\pE[(\sum_{\{u,u'\} \in E(G)} \sum_{v, v' : \{v,v'\} \in E(H)} \Pi[u \mapsto v] \Pi[u' \mapsto v'] -  |E(G)|) \cdot p^2] \geq 0$ for all polynomials $p$;
        \item $\pE[p^2] \geq 0$ for all polynomials~$p$.
    \end{enumerate}
We say there is degree-$r$ SOS refutation for the isomorphism between $G$ and $H$ if the pseudo-expectation operator with properties listed above does not exist.

\begin{remark}
It is equivalent to replace (\ref{item:SOS-GISO}) by ``$\pE[(\sum_{v, v' : \{v,v'\} \in E(H)} \Pi[u \mapsto v] \Pi[u' \mapsto v'] - 1) \cdot q] = 0$ for all $(u, u') \in E(G)$ and all polynomials $q$''.
\end{remark}

\fi

\ifnum\full=0
Due to space constraints, all the proofs directly related to SOS/Lasserre hierarchy are deferred to the full version of this paper. Therefore, here we also omit the definition of notations on SOS/Lasserre hierarchy.
\fi

\ifnum\full=1
\subsection{Random \threexor}
\fi

\ifnum\full=0
\vspace{-2ex}
\paragraph{Random \threexor.}
\fi

A random \threexor instance with $n$ variables and $m$ equations is sampled by choosing $m$ unordered 3-tuples of variables from all possible ${n \choose 3}$ ones, and making each 3-tuple $(x_{j_1}, x_{j_2}, x_{j_3})$ into a \threexor constraint $x_{j_1} + x_{j_2} + x_{j_3} = b$ with an independent random $b \in \Z_2$. 
\begin{theorem}\cite{Sch08}\label{thm:sch08}
For every constant $c > 1$, there is exists $\eta > 0$ such that with probability $1 - o(1)$, the satisfiability of a random \threexor instance\footnote{The random \threexor distribution used in \cite{Sch08} is slightly different, but the theorem still holds for our distribution.} with $n$ variables and $cn$ equations cannot be refuted by degree-$(\eta n)$ SOS/Lasserre hierarchy.
\end{theorem}
\ifnum\full=0
\vspace{-3ex}
\fi

\section{Reduction from \threexor to \giso} \label{sec:reduction-3xor}

\ifnum\full=0
\vspace{-1ex}
\fi

We define a slight variant of the basic gadget from~\cite{CFI92}:
\begin{definition}\label{def:gadgetgraph}
    Let $C$ be a \threexor constraint involving variables $x_1$, $x_2$, $x_3$.  The associated gadget graph $G_C$ consists of: $6$~``variable vertices'' with names ``$x_i \mapsto a$'' for each $i \in [3]$, $a \in \Z_2$; and, $4$~``constraint vertices'' with names ``$x_1 \mapsto a_1, x_2 \mapsto a_2, x_3 \mapsto a_3$'' for each partial assignment which satisfies the constraint~$C$.  Regarding edges, each pair of variable vertices $x_i \mapsto 0$, $x_i \mapsto 1$ is connected by an edge; the four constraint vertices are  connected by a clique; and, each constraint vertex~$\alpha$ is connected to the three variable vertices it is consistent with.
\end{definition}
    
\rnote{insert lemma about the automorphisms of this graph}

Now we describe how an entire instance of \threexor is encoded by a graph:
\begin{definition}\label{def:3xorencoding}
    Let $\calC$ be a collection of \threexor constraints over variable set~$\calX$.  We define the associated graph $G_{\calC}$ as follows:  For each constraint $C \in \calC$, the graph $G_{\calC}$ contains a copy of the gadget graph~$G_C$.  However we  \emph{identify} all of the variable vertices $x \mapsto a$ across $x \in \calX$, $a \in Z_2$ as well as the variable edges $(x \mapsto 0, x \mapsto 1)$.  The constraint vertices associated to~$C$, on the other hand, are left as-is, and will be named $\alpha_C$. We denote the set of vertices $\{x \mapsto 0, x \mapsto 1\}$ by $V_x$ for every variable $x$, denote the set of vertices corresponding to $C$ by $V_C$ for every variable $C$.
\end{definition}
\begin{remark}
    If $\calC$ is a \threexor instance with $n$ vertices and $m$ constraints then the graph $G_{\calC}$ has $N = 4m+2n$ vertices and $M = 18m + n$ edges.
\end{remark}

Finally, we introduce the following notation:
\begin{notation}
    Let $C$ be a \threexor constraint involving variables $x_i, x_j, x_k$.  We write $\homog{C}$ for its homogeneous version, $x_i+x_j+x_k = 0$.  Given a collection of \threexor constraints $\calC$ we write $\homog{\calC} = \{\homog{C} : C \in \calC\}$.
    
\paragraph{The reduction.} Given a collection of \threexor constraints $\calC$, the corresponding \giso instance i $(G_\calC, G_{\homog{\calC}})$.
\end{notation}

\ifnum\full=0
\vspace{-3ex}
\fi

\section{Completeness} \label{sec:completeness-main}
\ifnum\full=0
\vspace{-1ex}
\fi

\begin{lemma}[Completeness]\label{lem:completeness}
If $\calC$ is a \threexor instance such that $\val(\calC) \geq 1 - \epsilon$, then $\GI(G_\calC, G_{\homog{\calC}}) \geq 1 - 2 \epsilon/3$. 
\end{lemma}

\ifnum\full=1

\begin{proof}
Let $\tau$ be an assignment to the variables in $\calC$ such that $\val(\calI; \tau) \geq 1 - \epsilon$. Now we define a bijection $\pi$ from the vertices in $G_\calC$ to the ones in $G_{\homog{\calC}}$ as follows. 

For each variable vertex $x_j \mapsto b$, let $\pi(x_j \mapsto b) = x_j \mapsto b + \tau(x_j)$. For any equation vertex $\alpha_{C_i}$, if $C_i$ is not satisfied by $\tau$, map it to an arbitrary vertex in $V_{\homog{C_i}}$. If $C_i$ is satisfied by $\tau$, let us suppose that $C_i : x_{j_1} + x_{j_2} + x_{j_3} = b$, let $\alpha'$ be an assignment such that $\alpha'(x_{j_t}) = \alpha(x_{j_t}) + \tau(x_{j_t})$ for all $t \in \{1, 2, 3\}$. Observe that 
\[\alpha'(x_{j_1}) + \alpha'(x_{j_2}) + \alpha'(x_{j_3}) = (\alpha(x_{j_1}) + \alpha(x_{j_2}) + \alpha(x_{j_3})) + (\tau(x_{j_1}) + \tau(x_{j_2}) + \tau(x_{j_3})) = b + b = 0. \]
Therefore $\alpha'_{\homog{C_i}}$ is a vertex in $G_{\homog{\calC}}$. We let $\pi$ map $\alpha_{C_i}$ to $\alpha'_{\homog{C_i}}$. 

It is straightforward to check that all the edges between equation vertices and between variable vertices are satisfied. Now we consider an edge between a equation vertex and a variable vertex, namely between $\alpha_{C_i}$ and $x_j \mapsto b$ where $x_j$ is an variable in equation $C_i$ and $\alpha(x_j) = b$. We show that the edge is satisfied by $\pi$ whenever $C_i$ is satisfied by $\tau$. Let $\alpha'$ and $b'$ be such that $\pi(\alpha_{C_i}) = \alpha'_{\homog{C_i}}$, $\pi(x_j \mapsto b) = x_j \mapsto b'$. Observe that
\[
\alpha'(x_j) = \alpha(x_j) + \tau(x_j) = b + \tau(x_j) = b',
\] 
and this implies that there is an edge between  $\alpha'_{C_i}$ and $x_j \mapsto b'$.  

We see that the only edges in $G_{\calC}$ which might not be satisfied by $\pi$ are the ones between equation vertices and variable vertices where the corresponding equation vertex is not satisfied by $\tau$. For each equation not satisfied, there are at most $12$ such edges. Therefore there are at most $12\epsilon m$ edges not satisfied. We have 
\[
\GI(G_\calC, G_{\homog{\calC}}) \geq \GI(G_\calC, G_{\homog{\calC}}; \pi) \geq \frac{M - 12\epsilon m}{M} \geq 1 - \frac{2}{3} \epsilon .\qedhere
\]
\end{proof}

\subsection{SOS completeness}
\fi
\begin{lemma}[SOS completeness]\label{lem:sos-completeness}
    Let $\calC$ be a \threexor instance on variable set $\calX$ and suppose that every SOS refutation of~$\calC$ requires  degree exceeding~$r$.  Then every SOS refutation of the statement ``$G_{\calC}$ and $G_{\homog{\calC}}$ are isomorphic'' requires degree exceeding~$r/3$.
\end{lemma}

\ifnum\full=1
\begin{proof}
    Since $\calC$ cannot be refuted in degree~$r$, there is a pseudo-expectation operator $\pE_{\calC}[\cdot]$ defined on the space of real polynomials of degree at most~$r$ over the indeterminates $(A[x \mapsto a])_{x \in \calX, a \in \Z_2}$ associated to~$\calC$.  This $\pE_{\calC}[\cdot]$ is normalized, linear, and satisfies:
    \begin{enumerate}[i.]
        \item \label{pEC:zero-one} $\pE_{\calC}[(A[x \mapsto a]^2 - A[x \mapsto a]) \cdot q] = 0$ for all $x \in \calX$, $a \in \Z_2$, and polynomials~$q$;
        \item \label{pEC:assignment} $\pE_{\calC}[(A[x \mapsto 0] + A[x \mapsto 1] - 1) \cdot q] = 0$ for all $x \in \calX$ and polynomials~$q$;
        \item \label{pEC:constraints} $\pE_{\calC}[(\littlesum_{\alpha_C \text{\ satisfying $C$}} A[x_1 \mapsto \alpha_C(x_1)]A[x_2 \mapsto \alpha_C(x_2)]A[x_3 \mapsto \alpha_C(x_3)] - 1) \cdot q] = 0$ for each $C \in \calC$ involving variables $x_1, x_2, x_3$ and all polynomials~$q$;
        \item \label{pEC:SOS} $\pE_{\calC}[p^2] \geq 0$ for all polynomials~$p$.
    \end{enumerate}

    Our task is to define a normalized linear pseudo-expectation operator $\pE_{\calG}[\cdot]$ on the space of real polynomials of degree at most~$r/3$ over the indeterminates $(\Pi[u \mapsto v])_{u \in V(G_{\calC}), v \in V(G_{\homog{\calC}})}$ such that:
    \begin{enumerate}[I.]
        \item \label{pEG:zero-one} $\pE_{\calG}[(\Pi[u \mapsto v]^2 - \Pi[u \mapsto v]) \cdot q] = 0$ for all $u \in V(G_{\calC})$, $v \in V(G_{\homog{\calC}})$, and polynomials~$q$;
        \item \label{pEG:injection} $\pE_{\calG}[(\sum_{v \in V(G_{\homog{\calC}})} \Pi[u \mapsto v] - 1) \cdot q] = 0$ for all $u \in V(G_{\calC})$ and polynomials~$q$; and similarly, $\pE_{\calG}[(\sum_{u \in V(G_{{\calC}})} \Pi[u \mapsto v] - 1) \cdot q] = 0$ for all $v \in V(G_{\homog{\calC}})$ and polynomials~$q$;
        \item \label{pEG:isomorphism} $\pE_{\calG}[(\sum_{\{u,u'\} \in E(G_{\calC})} \sum_{v, v' : \{v,v'\} \in E(G_{\homog{\calC}})} \Pi[u \mapsto v] \Pi[u' \mapsto v'] -  M) \cdot p^2] \geq 0$ for all polynomials $p$;
        \item \label{pEG:SOS} $\pE_{\calG}[p^2] \geq 0$ for all polynomials~$p$.
    \end{enumerate}
    Here $M$ denotes the number of edges in $G_{{\calC}}$ (and also in $G_{\homog{\calC}}$).\\

    The idea is to formally define each indeterminate $\Pi[u \mapsto v]$ as a certain degree-$3$ multilinear polynomial of the indeterminates $A[x \mapsto a]$.  Then $\pE_{\calG}[\cdot]$ is automatically defined in terms of~$\pE_{\calC}[\cdot]$ for all polynomials of degree at most~$r/3$.  The natural definition is as follows:
    \begin{enumerate}
        \item Let $x \in \calX$ and $a \in \Z_2$.  We define $\Pi[(x \mapsto a) \mapsto (x \mapsto b)] = A[x \mapsto (a-b)]$.
        \item Let $C \in \calC$, let $\alpha_C = (x_1 \mapsto a_1, x_2 \mapsto a_2, x_3 \mapsto a_3)$ be constraint vertex in $G_{\calC}$ corresponding to~$C$, and let $\beta_{\homog{C}} = (x_1 \mapsto b_1, x_2 \mapsto b_2, x_3 \mapsto b_3)$ be a constraint vertex in $G_{\homog{\calC}}$ corresponding to $\homog{C}$. We define ${\Pi[\alpha_C \mapsto \beta_{\homog{C}}]}$ to be the following degree-$3$ monomial: 
        \[{A[x_1 \mapsto (a_1 - b_1)] A[x_2 \mapsto (a_2 - b_2)] A[x_3 \mapsto (a_3 - b_3)]}. \]
        \item All other indeterminates $\Pi[u \mapsto v]$ are formally defined to be~$0$.
    \end{enumerate}

It is clear that $\pE_{\calG}[\cdot]$ is normalized and linear by the same property of $\pE_{\calC}[\cdot]$. It remains to show that the induced pseudo-expectation operator $\pE_{\calG}[\cdot]$ satisfies \eqref{pEG:zero-one}--\eqref{pEG:SOS} using the fact that $\pE_{\calC}[\cdot]$ satisfies \eqref{pEC:zero-one}--\eqref{pEC:SOS}.  Most of these are easy; for example, the implication \eqref{pEC:SOS}~$\Rightarrow$~\eqref{pEG:SOS} is immediate.  Almost as easy is that \eqref{pEC:zero-one}~$\Rightarrow$~\eqref{pEG:zero-one} and that \eqref{pEC:zero-one},~\eqref{pEC:assignment}~$\Rightarrow$~\eqref{pEG:injection}. We illustrate some of these implications, leaving the rest to the reader. For example, let's verify~\eqref{pEG:zero-one} for indeterminates of type $\Pi[\alpha_C \mapsto \beta_{\homog{C}}]$. For brevity we'll write $\Pi[\alpha_C \mapsto \beta_{\homog{C}}]$ as $A_1A_2A_3$.  Now for any polynomial~$q$ over the $\Pi$'s,
    \begin{align*}
        &\phantom{{}={}} \pE_{\calG}[(\Pi[\alpha_C \mapsto \beta_{\homog{C}}]^2 - \Pi[\alpha_C \mapsto \beta_{\homog{C}}]) \cdot q] \\
        &= \pE_{\calG}[(A_1^2A_2^2A_3^2 - A_1A_2A_3)\cdot q'] \tag{for some polynomial $q'$ over the $A$'s} \\
        &= \pE_{\calG}[(A_1^2 - A_1)A_2^2A_3^2\cdot q'] + \pE_{\calG}[A_1(A_2^2-A_2)A_3^2\cdot q'] + \pE_{\calG}[A_1A_2(A_3^2-A_3)\cdot q']\\
        &= 0 \tag{by~\eqref{pEC:zero-one}.}
    \end{align*}
    And let's verify~\eqref{pEG:injection} when $u$ is a variable vertex $x \mapsto a$:
    \begin{align*}
        &\phantom{{}={}} \pE_{\calC}\left[\left(\littlesum_{v \in V(G_{\homog{\calC}})} \Pi[(x \mapsto a) \mapsto v] - 1\right) \cdot q\right] \\
        &= \pE_{\calC}[(A[x \mapsto a-0] + A[x \mapsto a-1] - 1) \cdot q'] \tag{all other $\Pi$'s are $0$} \\
        &= 0 \tag{by \eqref{pEC:assignment} }.
    \end{align*}

    The main effort is to establish~\eqref{pEG:isomorphism}.  In fact we will show
    \begin{equation}    \label{eqn:edge-to-edge}
        \pE_{\calC}\left[\left(\littlesum_{v, v' : \{v,v'\} \in E(G_{\homog{\calC}})} \Pi[u \mapsto v] \Pi[u' \mapsto v'] -  1\right) \cdot p^2\right] = 0
    \end{equation}
    for all edges $\{u,u'\} \in E(G_{\calC})$ and all~$p$, whence~\eqref{pEG:isomorphism} follows by summing.  We will omit the (easy) verification of this for the edges $(x \mapsto 0, x \mapsto 1)$.  Instead we will first verify that~\eqref{eqn:edge-to-edge} holds for a typical clique edge associated to constraint~$C$, say $(\alpha_C, \alpha'_C)$, on variables $x_1, x_2, x_3$. Only the indeterminates of corresponding constraints, say $\Pi[\alpha_C \mapsto \beta_{\homog{C}}]$, are nonzero.  Writing $A_i[\alpha - \beta] = A[x_i \mapsto \alpha_C(x_i) - \beta_{\homog{C}}(x_i)]$ for brevity (and similarly with primes), the quantity in~\eqref{eqn:edge-to-edge} is
    \begin{align}
        &\phantom{{}={}} \pE_{\calC}\Bigl[\Bigl(\sum_{\substack{\beta_{\homog{C}}, \beta'_{\homog{C}} \\ \text{satisfying $\homog{C}$}}} A_1[\alpha - \beta]A_2[\alpha - \beta]A_3[\alpha - \beta]A_1[\alpha' - \beta']A_2[\alpha' - \beta']A_3[\alpha' - \beta'] -  1\Bigr) \cdot {p'}^2\Bigr] \nonumber\\
        &= \pE_{\calC}\Bigl[\Bigl(\bigl(\littlesum_{\beta_{\homog{C}}} A_1[\alpha - \beta]A_2[\alpha - \beta]A_3[\alpha - \beta]\bigr)\bigl(\littlesum_{\beta'_{\homog{C}}}A_1[\alpha' - \beta']A_2[\alpha' - \beta']A_3[\alpha' - \beta']\bigr) -  1\Bigr) \cdot {p'}^2\Bigr]. \label{eqn:i'm-0}
    \end{align}
    Now for fixed $\alpha_C$, as $\beta_{\homog{C}}$ ranges over all satisfying assignments to $\homog{C}$, the assignment $\alpha_C - \beta_{\homog{C}}$ ranges over all satisfying assignments to~$C$.  The analogous statement holds also for~$\alpha'_C$.  It's now straightforward to see that the vanishing of~\eqref{eqn:i'm-0} follows from~\eqref{pEC:constraints}.\rnote{ (subtract and add~$1$ from $\littlesum_{\beta_{\homog{C}}} A_1[\alpha - \beta]A_2[\alpha - \beta]A_3[\alpha - \beta]$ and expand; repeat with $\alpha'$ and $\beta'$).}  
    
    Our final task is to verify~\eqref{eqn:edge-to-edge} also for edges between variable vertices and constraint vertices.  Fix a typical such edge, say  one connecting $x_1 \mapsto a_1$ to $\alpha_C$.  (We'll use the same notation as before for~$\alpha_C$; in particular, note that we must have $\alpha_C(x_1) = a_1$.) Now in this case, the quantity in~\eqref{eqn:edge-to-edge} is
    \begin{align}
        &\phantom{{}={}} \pE_{\calC}\Bigl[\Bigl(\sum_{\substack{b \in \Z_2, \\ \beta_{\homog{C}}\text{ satisfying $\homog{C}$}}} A[x_1 \mapsto (a_1-b)] A_1[\alpha - \beta]A_2[\alpha - \beta]A_3[\alpha - \beta] -  1\Bigr) \cdot {p'}^2\Bigr] \nonumber\\
        &= \pE_{\calC}\Bigl[\Bigl(\bigl(\littlesum_{c \in \Z_2} A[x_1 \mapsto c]\bigr)\bigl(\littlesum_{\beta_{\homog{C}}} A_1[\alpha - \beta]A_2[\alpha - \beta]A_3[\alpha - \beta]\bigr) -  1\Bigr) \cdot {p'}^2\Bigr]. \label{eqn:i'm-also-0}
    \end{align}
    Again, the fact that~\eqref{eqn:i'm-also-0} vanishes now easily follows from~\eqref{pEC:assignment},~\eqref{pEC:constraints}.
\end{proof}

\fi

\ifnum\full=0
The proofs are deferred to the full version of this paper due to space constraints. 
\fi

\ifnum\full=0
\vspace{-2ex}
\fi

\section{Soundness} \label{sec:soundness-main}

\ifnum\full=1
In this section, we prove the following soundness lemma.
\fi

\begin{lemma}[Soundness]\label{lem:soundness}
Let $\calC = \{C_1, C_2, \dots, C_m\}$ be a random \threexor instance with $n$ variables and $m = cn$ equations where $c \geq 10^10$. With probability $1 - o(1)$, we have
\ifnum\full=1
\[
\fi
\ifnum\full=0
$
\fi
\GI(G_C, G_{\homog{C}}) < 1 - \frac{1}{95 c^2} . 
\ifnum\full=1
\]
\fi
\ifnum\full=0
$
\fi
\end{lemma}

Before proving \pref{lem:soundness}, we first introduce the following definition.

\begin{definition}
A graph (possibly hypergraph) $G$ is \emph{$(\epsilon, D)$-degree bounded} if the average degree of every set of $\epsilon$ fraction of vertices is at most $D$.
\end{definition}

\begin{claim}\label{claim:degree-bounded}
Suppose $c \geq 3$. A random 3-uniform hypergraph $H$ drawn from $\calG^{(3)}_{n, m}$, where $m = cn$, is $(1/c, 100c)$-degree bounded with probability $1 - o(1)$.
\end{claim}

\ifnum\full=1
\begin{claim}\label{claim:bounded-degree-set}
Given an $(\epsilon, D)$-degree bounded graph $G$ with $n$ vertices, every set of $\beta n$ vertices has at most $(\epsilon+\beta) Dn$ edges incident to them.
\end{claim}
\fi

\pref{lem:soundness} is directed implied by the following two lemmas.
\begin{lemma}\label{lem:soundness-decoding}
Let $H = ([n], E = \{e_i\})$ be a $3$-uniform hypergraph with $n$ vertices and $m = cn$ hyperedges. the constraint graph of a \threexor instance with $n$ variables and $m = cn$ equations. Suppose $H$ is $(\epsilon, 100c)$-degree bounded and $(\beta, \gamma)$-asymmetric, where $\gamma \geq 200\epsilon$. Let $\calC = \{C_1, C_2, \dots, C_m\}$ be an arbitrary \threexor with $n$ variables and $m$ constraints based on $H$. (In other words, each hyperedge $e_i$ of $H$ connects the indices of the $3$ variables used by $C_i$.)   If we set
\ifnum\full=1
\[
\fi
\ifnum\full=0
$
\fi
\delta := \delta(c, \epsilon, \beta, \gamma) = \min\left\{\frac{1}{200}, \frac{\gamma}{48}, \frac{\epsilon}{95c}\right\},
\ifnum\full=1
\]
\fi
\ifnum\full=0
$
\fi
when $\GI(G_C, G_{\homog{C}})  \geq 1 - \delta$, we have $\val(\calC) \geq .9 - 100(\epsilon+ \beta)$. 
\end{lemma}

\begin{lemma}\label{lem:random-threelin-soundness}
If $\calC$ is a random \threexor instance with $n$ variables and $m \geq 10000n$ equations, then with probability $1 - o(1)$, we have $\val(\calC) < .51$. 
\end{lemma}

\begin{proof}[Proof of \pref{lem:soundness}]
Set $\epsilon = \frac{1}{c}$, $\gamma = \frac{200}{c}$, $\beta = \frac{48000}{c}$.  Combining \pref{claim:degree-bounded}, \pref{lem:random-threelin-soundness}, and \pref{thm:randomhypergraphwhp}, we know that with probability $(1-o(1))$, all of the following hold:
\begin{enumerate}
\item $H$ is $(\eps, 100c)$-degree bounded,\label{item:degreeboundorig}
\item $H$ is $(\beta, \gamma)$-asymmetric, and \label{item:asymmetricorig}
\item $\val(\calC) < .51$.\label{item:valueorig}
\end{enumerate}
Given that these hold, assume for sake of contradiction that $\GI(G_\calC, G_{\homog{\calC}}) \geq 1 - \frac{1}{95c^2}$.  Then because $G$ satisfies \pref{item:degreeboundorig} and \pref{item:asymmetricorig}, \pref{lem:soundness-decoding} implies that
\begin{equation*}
\val(\calC) \geq .9 - 100\left(\frac{1}{c} + \frac{48000}{c}\right) \geq .8,
\end{equation*}
where the last step follows because $c \geq 10^{10}$.  However, this contradicts \pref{item:valueorig}.  Therefore, $\GI(G_\calC, G_{\homog{\calC}}) < 1 - \frac{1}{95c^2}$ with probability $1-o(1)$.
\end{proof}

\ifnum\full=1
The proof of \pref{lem:random-threelin-soundness} is standard. 
\begin{proof}[Proof of \pref{lem:random-threelin-soundness}]
Fix an assignment to the $n$ variables. The probability that the assignment satisfies at least $.51m$ equations of a random \threexor instance is at most $\exp(-.02^2 \cdot .5m/2) = \exp(-.0001 m)$ by the Chernoff bound. Since there are only $2^n$ assignments, the probability that no assignment satisfies more than $.5m$ equations is at least $1 - \exp(-.0001 m) \cdot 2^n = 1 - o(1)$ when $m \geq 10000n$.
\end{proof}
\fi

The rest of the section is devoted to the proof of \pref{lem:soundness-decoding}.

\begin{proof}[Proof of \pref{lem:soundness-decoding}]
Let $\pi$ be a bijection mapping the vertices in $G_\calC$ to the vertices in $G_{\homog{\calC}}$ such that $\GI(G_\calC, G_{\homog{\calC}}; \pi) \geq 1 - \delta$. We first prove that for most $i$'s, $\pi$ maps the set $V_{C_i}$ to $V_{\homog{C_{i'}}}$ for some $i'$, and for most $j$'s, $\pi$ maps most $V_{x_j}$ to $V_{x_j'}$ for some $j'$. Formally, let $A$ be the set of $i \in [m]$ such that $\pi(V_{C_i}) = V_{\homog{C_{i'}}}$ for some $i'$, and let $B$ be the set of $j \in [n]$ such that $\pi(V_{x_j}) = V_{x_{j'}}$ for some $j'$. We show that 

\begin{claim}\label{claim:soundness-correspondence}
$|A| \geq (1 - 19\delta)m$, $|B| \geq (1 - 95c\delta) n$. 
\end{claim}

Now we are able to define a permutation $\sigma$ on the variables in $\calC$ (as well as $\homog{\calC}$ since the set of variables is shared). We let $\sigma$ to be an arbitrary permutation so that for each $j \in B$, we have $\sigma(j) = j'$ where $\pi(V_{x_j}) = V_{x_{j'}}$. Now we show that $\sigma$ is an almost automorphism for the constraint graph (i.e. the hypergraph $H = ([n], E)$). 

\begin{claim}\label{claim:soundness-almost-automorphism}
$\aut(H; \sigma) \geq (1 - 100 \epsilon - 24\delta) m$.
\end{claim}
By our setting of $\delta$, we have $24\delta < \gamma/2$. Since we also assume that $100 \epsilon  < \gamma/2$, we have $\aut(G; \sigma) \geq (1 - \gamma) m$, and therefore we know that $\sigma$ has at least $(1 - \beta) n$ fixed points.

Now we are ready to define an assignment $\tau : \{x_j\} \to \Z_2$ which certifies that $\val(\calI) \geq .9$. For each $j$ which is not a fixed point of $\sigma$, define $\tau(x_j)$ arbitrarily. For each $j$ being a fixed point of $\sigma$, we know that $\pi(V_{x_j}) = V_{x_j}$. We let $\tau(x_j) = b$ where $ \pi(x_j \mapsto 0) = x_j \mapsto b$. We conclude the proof by showing the following claim.
\begin{claim}\label{claim:soundness-threelin-val}
$\val(\calC; \tau) \geq .9 - 100(\epsilon+ \beta)$.
\end{claim}
\ifnum\full=0
\vspace{-5ex}
\fi
\end{proof}

\ifnum\full=0
The proofs of the claims are deferred to the full version of this paper due to space constraints. 
\fi

\ifnum\full=1

\subsection{Proof of the claims}

\begin{proof}[Proof of \pref{claim:soundness-correspondence}]
Observe that the only $4$-cliques in $G_{\homog{\calC}}$ are $V_{\homog{C_{i'}}}$ ($i' \in [m]$). Therefore, if $\pi(V_{C_i}) \neq V_{\homog{C_{i'}}}$ for every $i'$, we know that at least one of the edges in the clique $V_{C_i}$ is not satisfied. Therefore we have $m - |A| \leq \delta M$ (recall that $M = 18m + n$ is the number of edges in $G_{\calC}$), i.e. $|A| \geq m - \delta M \geq (1 - 19\delta) m$. 

We now know that at least $(1 - 19\delta)m \cdot 4$ equation vertices in $G_{\homog{\calC}}$ are mapped from equation vertices in $G_\calC$. Therefore there are at most $4m - (1 - 19\delta)m \cdot 4 = 76\delta m$ equation vertices in $G_{\homog{\calC}}$ being mapped from variable vertices in $G_\calC$. In other words, $\pi$ maps at least $2n - 76\delta m = (2 - 76c\delta)n$ variable vertices to variable vertices. Let $B'$ be the set of $j$'s such that both vertices in $V_{x_j}$ is mapped to a variable vertex. We have $B' \supseteq B$ and $|B'| \geq (2 - 76c\delta)n - n = (1 - 76c\delta) n$. For each $j \in B' \setminus B$, we know that the edge in $V_{x_j}$ is not satisfied. Therefore $|B' \setminus B| \leq \delta M$. Therefore, $|B| = |B'| - |B' \setminus B| \geq (1 - 76c \delta) n - \delta M \geq (1 - 95c\delta) n$. 
\end{proof}

\begin{proof}[Proof of \pref{claim:soundness-almost-automorphism}]
Let $E'$ be the set of hyperedges in $E$ whose vertices are all in $B$. Since $G$ is $(\epsilon, 100c)$-degree bounded, and by our setting of parameters $95c\delta < \epsilon$, by \pref{claim:bounded-degree-set}, we know that $|E'| \geq m - 100c \epsilon n = (1 - 100 \epsilon ) m$. Now let us consider the hyperedges in $E'' = E' \cap A$. (Also observe that $|E''| \geq (1 - 100 \epsilon - 19\delta) m$.)  We claim that most of the hyperedges in $E''$ are satisfied by $\sigma$. For every hyperedge $e_i = \{j_1, j_2, j_3\} \in E''$ that is not satisfied, we know that $\{\sigma(j_1), \sigma(j_2), \sigma(j_3)\} \not\in T$. Since $i \in E'' \subseteq A$, let $i'$ be the equation index such that $\pi(V_{C_i}) = V_{\homog{C_{i'}}}$. Since $\{\sigma(j_1), \sigma(j_2), \sigma(j_3)\} \not\in E$, we have $e_{i'} \neq \{\sigma(j_1), \sigma(j_2), \sigma(j_3)\}$. Let us assume w.l.o.g. that $\sigma(j_1) \not\in e_{i'}$. Then there is no edge between $V_{\homog{C_{i'}}}$ and $V_{x_{\sigma(j_1)}}$ in $G_{\homog{\calC}}$. Therefore the $4$ edges between $V_{C_i}$ and $V_{x_{j_1}}$ in $G_\calC$ are not satisfied. 

We have proved that whenever there is an hyperedge in $E''$ not satisfied by $\sigma$, there are at least $4$ edges in $G_\calC$ not satisfied by $\pi$. Since $\pi$ satisfies $(1 - \delta)M$ edges, there are at most $\delta M / 4$ hyperedges in $E''$ not satisfied by $\sigma$. Therefore, we have 
\[
\aut(H; \sigma) \geq |E''| - \delta M / 4 \geq (1 - 100 \epsilon - 19\delta - 5\delta)m = (1 - 100\epsilon - 24\delta) m .\qedhere
\]
\end{proof}

\begin{proof}[Proof of \pref{claim:soundness-threelin-val}]
Let $E'$ be the set of hyperedges in $E$ whose vertices are all fixed points of $\sigma$. Since $\sigma$ has at least $(1 - \beta) n$ fixed points, and $H$ is $(\epsilon, 100c)$-degree bounded, by \pref{claim:bounded-degree-set}, we know that $|E'| \geq m - (\epsilon+ \beta) \cdot 100cn = (1 - 100 (\epsilon+ \beta))m$. Now consider any $e_i \in E'$, if all the edges incident to $V_{C_i}$ are satisfied, let $E''$ contain $i$. Since there are at most $\delta M \leq 19\delta m$ edges not satisfied, we know that $|E''| \geq |E'| - 19\delta m \geq (1 - 100 (\epsilon+ \beta) - 19\delta)m$. We claim that for all $e_i \in E''$, the equation $C_i$ is satisfied. Therefore we have $\val(\calI; \tau) \geq 1 - 100 (\epsilon+ \beta) - 19\delta \geq .9 - 100 (\epsilon+ \beta)$, since $\delta < 1/200$.

Now we show that $C_i$ is satisfied by $\tau$ when $e_i \in E''$. Using similar argument in the proof of \pref{claim:soundness-almost-automorphism}, one can show that when $e_i = \{j_1, j_2, j_3\} \in E''$, we have $\pi(V_{C_i}) = V_{\homog{C_i}}$. Also we have $\pi(V_{x_{j_t}}) = V_{x_{j_t}}$ for all $t \in \{1, 2, 3\}$, by the definition of $E''$. Let $H$ be the induced subgraph $G_\calC[V_{C_i} \cup (\cup_{t \in \{1, 2, 3\}} V{x_{j_t}})]$, let $\homog{J}$ be the induced subgraph $G_{\homog{\calC}}[{V_{\homog{C_i}}} \cup (\cup_{t \in \{1, 2, 3\}} V_{x_{j_t}})]$. We use the following claim to conclude the proof.

\begin{claim}\label{claim:soundness-gadget}
If $\pi$ (after projected on the suitable vertices) is an isomorphism between $J$ and $\homog{J}$, $C_i$ is satisfied by $\tau$.
\end{claim}
\end{proof}

It remains to prove \pref{claim:soundness-gadget}. We first claim the following property about our construction $\calL(\cdot)$.
\begin{claim}\label{claim:soundness-gadget-helper}
Let $\calC$ be a \threexor instance, let $C : x_{j_1} + x_{j_2} + x_{j_3} = b$ be an equation from $\calC$. For any $b_1, b_2, b_3 \in \Z_2$, and any vertex $\alpha_{C}$, the parity of the number of neighbors of $\alpha_C$ in $\{x_{j_1} \mapsto b_1, x_{j_2} \mapsto b_2, x_{j_3} \mapsto b_3\}$ is $b + b_1 + b_2 + b_3$ .
\end{claim}

Now we are ready to prove \pref{claim:soundness-gadget}.
\begin{proof}[Proof of \pref{claim:soundness-gadget}]
Suppose the equation $C_i$ is $x_{j_1} + x_{j_2} + x_{j_3} = b$. 
Consider $\alpha_{C_i} \in V_{C_i}$, let $\alpha'_{\homog{C_i}} = \pi(\alpha_{C_i})$. By the construction of $G_\calC$ and $G_{\homog{\calC}}$, we know that $\alpha(x_{j_1}) + \alpha(x_{j_2}) + \alpha(x_{j_3}) = b$, and $\alpha'(x_{j_1}) + \alpha'(x_{j_2}) + \alpha'(x_{j_3}) = 0$. Now let the set $A = \{x_{j_1} \mapsto 0, x_{j_2} \mapsto 0, x_{j_3} \mapsto 0\}$. By \pref{claim:soundness-gadget-helper}, we know that the parity of the number of neighbors of $\alpha_{C_i}$  in $A$ is $b + 0 + 0 + 0 = b$. Therefore, by isomorphism, the parity of the number of neighbors of $\alpha'_{\homog{C_i}}$ in $\pi(A)$ is also $b$. On the other hand, by the definition of $\tau$, we know that $\pi(A) = \{x_{j_1} \mapsto \tau(x_{j_1}), x_{j_2} \mapsto \tau(x_{j_2}), x_{j_3} \mapsto \tau(x_{j_3})\}$. By \pref{claim:soundness-gadget-helper} again, we know that the parity of the number of neighbors of $\alpha'_{\homog{C_i}}$ in $\pi(A)$ is $0 + \tau(x_{j_1}) + \tau(x_{j_2}) + \tau(x_{j_3})$. Therefore, we have that $\tau(x_{j_1}) + \tau(x_{j_2}) + \tau(x_{j_3}) = b$, i.e. $C_i$ is satisfied by $\tau$. 
\end{proof}

\fi
\ifnum\full=0
\vspace{-2ex}
\fi

\section{Random graphs are robustly asymmetric} \label{sec:robust-asymmetry}

In this section we prove \pref{thm:robust-asymmetry}.

We first set up some definitions.   For any graph $G=(V,E)$, let $\pi$ be a permutation over the vertices in $V$, we write $\id(\pi)$ as the number of fixed points in the permutations, that is, $\id(\pi)=|\{v\in V:\pi(v)=v\}|$. We define $\diff(G,\pi(G))=\{e:e\in E, \pi(e)\not\in E\}\bigcup \{e: e\not\in E, \pi(e)\in E\}$. Note that $\aut(G; \pi) =|E| - \frac{1}{2}|\diff(G,\pi(G))|$.

 For any permutation $\pi$ over the vertex set $V$, we define a directed graph $G_{\pi}=(\binom{V}{2},E_{\pi})$, and $(e_1,e_2)\in E_{\pi}$ if and only if $e_2=\pi(e_1)$. Since each  $e=\{u,v\} \in {V \choose 2}$ has in-degree and out-degree exactly 1, we can divide $G_{\pi}$ into disjoint unions of directed cycles. We call each directed cycle a bin, and the size of the bin is the number of elements in the cycle.
 

\begin{fact}
    For any size-1 bins, there are only two situations:
    \begin{itemize}
        \item $e=\{u,v\}$ where $u$ and $v$ are both fixed points of $\pi$. We call these bins \emph{type-1} size-1 bins. The number of type-1 size-1 bins is at most $\binom{\id(\pi)}{2}$.
        \item $e=\{u,v\}$ where $\pi(u)=v$ and $\pi(v)=u$. We call these classes \emph{type-2} size-1 bins. The number of type-2 size-1 bins is at most $\frac{n-\id(\pi)}{2}$ for any permutation $\pi$.
    \end{itemize}
\end{fact}

Now let us consider $G \sim \calG_{n, m}$ where $m = cn$. Let $Z = {n \choose 2}$ be the number of possible edges from which we choose $m$ edges. We first prove the following lemma.

\begin{lemma} \label{lem:balancededgesimproved}
    Let $\mathcal{A}$ be the event that for any permutation $\pi$ such that $\id(\pi)= (1-\beta) n$, the number of the  edges in $G$ that fall into the bins of size $\geq 2$ is at least $\beta m /60$ and at most $2\beta Z/10^9$. Whenever $\beta \geq \exp(-c/6)$, we have
\ifnum\full=1
\[
\fi
\ifnum\full=0
$
\fi
        \Pr_{G\sim \mathcal{G}_{n,m}}[\mathcal{A}] = 1- n^{-\omega(1)}.
\ifnum\full=1
\]
\fi
\ifnum\full=0
$
\fi
\end{lemma}

By a union bound over all the $\beta : \beta \geq \beta^*$ (where there are at most $n$ of them), we get the following corollary.
\begin{corollary} \label{cor:balanced-edges}
 Let $\mathcal{B}$ be the event that for every $\beta \geq \beta^*$, any permutation $\pi$ such that $\id(\pi)= (1-\beta) n$, the number of the  edges in $G$ that fall into the bins of size $\geq 2$ is at least $\beta m /60$ at at most $2\beta Z/10^9$. We have
\ifnum\full=1
\[
\fi
\ifnum\full=0
$
\fi
        \Pr_{G\sim \mathcal{G}_{n,m}}[\mathcal{B}] = 1- n^{-\omega(1)}.
\ifnum\full=1
\]
\fi
\ifnum\full=0
$
\fi
\end{corollary}

Before we prove \pref{lem:balancededgesimproved}, we need the following lemma.

\begin{lemma}  \label{lem:incident}
    Let $G \sim \calG_{n,m}$. Suppose that $\beta \geq \beta^*$. With probability $1-n^{-\omega(1)}$, for any $T\subseteq V$, $|T|= \beta n$, the number of edges incident to $T$ is at least least $c \beta n / 40$.
\end{lemma}

Now we prove this lemma.

For any vertex $v\in V$, its expected degree in $\calG_{n, m}$ is $2c$. We would like to prove that the probability that the degree is at most $c/10$ is very low. Indeed, we claim a more general statement.

 \begin{claim}\label{claim:deg-prob}
Let $W$ be a set of $w$ possible edges from ${V \choose 2}$, where $\floor{(n-1) / 2} \leq w \leq n$,
\ifnum\full=1
\[
\fi
\ifnum\full=0
$
\fi
        \Pr_{G = (V, E) \sim \calG_{n,m}}[ |E \cap W| \leq c/10] \leq \exp(-c/2).
\ifnum\full=1
\]
\fi
\ifnum\full=0
$
\fi
 \end{claim}
Observe that when $W$ is the set of possible edges incident to $v$, \pref{claim:deg-prob} says that $\Pr[\deg(v) \leq c/10] \leq \exp(-c/3)$.

For each possible edge $e$, we define a random variable $X_e$ as the indicator variable for the event that $e$ is selected as an edge in $G$. \pref{claim:deg-prob} would be a direct application of Chernoff bound if the $X_e$ variables were independent. However, the following claim states that Chernoff bound still holds since the variables are negatively associated. (Please refer to e.g. \cite{JP83} for the definition of negatively associated random variables.)

\begin{claim}
The $Z$ variables $X_e$ are negatively associated.
\end{claim}
\begin{proof}
Since $\{X_e\}$ follows the permutation distribution over $m$ 1's and $(Z - m)$ 0's, by Theorem 2.11 in \cite{JP83}, the claim holds.
\end{proof}
\ifnum\full=0
Now \pref{claim:deg-prob} follows easily from Chernoff bound for negatively associated random variables, and so does the following claim.
\fi
\ifnum\full=1
\begin{proof}[Proof of \pref{claim:deg-prob}]
Since $\E[X_e] = 2c/(n-1)$ for every $e$, we have $\E[|E \cap W|] = \sum_{e \in W} \E[X_e] = 2c|W|/(n-1) \geq  c$. The claim follows by Chernoff bound for negatively associated variables.
\end{proof}

Similarly we show that
\fi
\begin{claim}\label{claim:incident-upperbound}
Let $W$ be a set of possible edges, when $c \leq n/10^{10}$,
\ifnum\full=1
\[
\fi
\ifnum\full=0
$
\fi
\Pr_{G = (V, E) \sim \calG_{n, m}} [|E \cap W| \geq |W|/10^9] \leq \exp(-3|W|/10^{10}).
\ifnum\full=1
\]
\fi
\ifnum\full=0
$
\fi
\end{claim}
\ifnum\full=1
\begin{proof}
Since $\E[X_e] = 2c/(n-1)$ for every $e$, we have $\E[|E \cap W|] = \sum_{e \in W} \E[X_e] = 2c|W|/(n-1) < 3 |W|/10^{10}$. The claim follows by Chernoff bound for negatively associated variables.
\end{proof}
\fi

Now we are ready to prove \pref{lem:incident}.
\begin{proof}[Proof of \pref{lem:incident}] Suppose the vertices of the random graph $G \sim \calG_{n, m}$ are numbered from $1$ to $n$.  Let $X_i = \sum_{j = i + 1}^{i + [(n-1)/2]} X_{e = \{i, (j-1)\mod n + 1\}}$. By \pref{claim:deg-prob}, we know that
\ifnum\full=1
\[
\fi
\ifnum\full=0
$
\fi
        \Pr[X_i\leq c/10] \leq \exp(-c/2).
\ifnum\full=1
\]
\fi
\ifnum\full=0
$
\fi

Since the random variables $\{X_i\}$ are sums of disjoint sets of negatively associated random variables $X_e$'s, we know that the $X_i$'s are also negatively associated. Let $U$ be the set of vertices $i$ such that $X_i \leq c/10$. We have $\E[|U|] \leq n \cdot \exp(-c/2)$. Using Chernoff bound for negatively associated random variables, we have
\ifnum\full=1
\[
\fi
\ifnum\full=0
$
\fi
\Pr\left[|U| \geq \frac{1}{2}\beta n\right] \leq \exp\left(-\frac{1}{3} \cdot \left(\frac{\beta n / 2}{\E[|U|]} - 1\right)^2\E[|U|]\right) = \exp\left(-\frac{1}{3} \cdot \left(\frac{\beta n / 2}{\E[|U|]} - 1\right) \left(\frac{\beta n}{2} - \E[|U|]\right)\right).
\ifnum\full=1
\]
\fi
\ifnum\full=0
$
\fi
Using $\beta \geq \exp(-c/6)$ and $c \geq 10$, we have $\beta n / 2 - \E[|U]] \geq \beta n /4$. Therefore,
\ifnum\full=1
\begin{multline*}
\Pr\left[|U| \geq \frac{1}{2}\beta n\right] \leq \exp\left(-\frac{1}{3} \left(\frac{\beta}{2} \cdot \exp(c/2) - 1\right) \cdot \frac{\beta n}{4}\right) \\ \leq \exp\left(-\frac{1}{3} \left(\frac{1}{2} \cdot \exp(c/3) - 1\right) \cdot \frac{\beta n}{4}\right) \leq \exp\left(-\frac{1}{3} \cdot \frac{1}{4} \cdot \exp(c/6) \cdot \frac{\beta n}{4}\right) \leq \exp\left(-\frac{n}{48}\right) = n^{-\omega(1)} , 
\end{multline*}
\fi
\ifnum\full=0
$\Pr\left[|U| \geq \frac{1}{2}\beta n\right]$ $ \leq \exp\left(-\frac{1}{3} \left(\frac{\beta}{2} \cdot \exp(c/2) - 1\right) \cdot \frac{\beta n}{4}\right) $  $\leq \exp\left(-\frac{1}{3} \left(\frac{1}{2} \cdot \exp(c/3) - 1\right) \cdot \frac{\beta n}{4}\right) $  $\leq \exp\left(-\frac{1}{3} \cdot \frac{1}{4} \cdot \exp(c/6) \cdot \frac{\beta n}{4}\right) $ $\leq \exp\left(-\frac{n}{48}\right) = n^{-\omega(1)} , $
\fi
where the second and fourth inequalities are because of $\beta \geq \exp(-c/6)$, and the third inequality is because of $\exp(c/3)/2-1 \geq \exp(c/6)/4$ for $c \geq 10$.

Therefore, with probability $1 - n^{-\omega(1)}$, there are at most $\beta n / 2$ vertices with degree at most $c/10$ (since $\deg(i) \geq X_i$ for every vertex $i$). When this happens, for any $T \subseteq V$, $|T| = \beta\cdot n$, there are at least $(|T| - \beta n / 2)$ vertices in $T$ with degree at least $c/10$, the sum of degrees of vertices in $T$ is at least $(|T|-\beta n/2) \cdot c / 10 = c\beta n/20$ , which means the number of edges incident to any vertex in $T$ is at least  $c\beta n/40$.
\end{proof}

\begin{proof} [Proof of \pref{lem:balancededgesimproved}]
By \pref{lem:incident}, we know that with probability $(1-n^{-\omega(1)})$, the number of edges in $G$ that is incident to $T$ is at least $c\beta n/40$, for every $T \subseteq V$ and $|T| = \beta n$. Therefore, for any $\pi$ with $\id(\pi)= (1-\beta) n$, let $T^*$ be the non-fixed points of $\pi$. We have $|T^*|= \beta n$. As the number of edges in size-1 bins which are incident to $T$ is at most $|T^*|/2 = \beta n /2$, the number of selected edges that in bins of size $\geq 2$ is at least $c\beta n / 40 -\beta n / 2\geq \beta m / 60$, when $c\geq 100$.

We also need to show that with probability $(1 - n^{-\omega(1)})$, the number of selected edges in bins of size $\geq 2$ (denote this number by random variable $X$) is at most $2 \beta^*Z / 10^9$. For every $\pi$ such that $\id(\pi) = (1 - \beta) n$, let $W$ be the set of possible edges whose end vertices are not both fixed point of $\pi$. We have $|W| = Z - {(1 -\beta) n \choose 2}$, therefore $\beta Z \leq W \leq 2\beta Z$ (for large enough $n$) and $X \leq |E \cap W|$. By \pref{claim:incident-upperbound} we have $\Pr[|E \cap W| \geq |W| / 10^9] \leq \exp(-3|W|/10^{10}) \leq \exp(-3\cdot \beta Z/10^{10})$. Therefore $\Pr[X \geq 2\beta Z / 10^9]  \leq \exp(-3\cdot \beta Z/10^{10}) \leq n^{-\omega(\beta n)}$.  We conclude the proof by taking a union bound over ${n \choose \beta n} < n^{\beta n}$ ways of choosing non-fixed points for $\pi$.
\end{proof}

\begin{lemma}       \label{lem:conditionedimproved}
Conditioned on event $\mathcal{B}$,  for $10^4 \leq c \leq n/10^7$, $\beta_0 \geq \beta^*$, with probability $(1-n^{-17})$, $G$ is $(\beta_0,\beta_0/240)$-asymmetric.
\end{lemma}
\begin{proof}

For any permutation $\pi$, we define a set of more fine-grained bins. We start with the bins we defined before, and split the bins of size $\geq 4$ into bins of size $2$ and at most one bin of size $3$ as follows. Suppose the original bin contains $\{e_1,e_2,\ldots,e_l\}$, where $\pi(e_{i})=e_{i+1}$ and $\pi(e_l)=e_1$, $l\geq 4$, the we have new bins which contains $\{e_1,e_2\},\ldots ,\{e_{l-1},e_{l}\}$ if $l$ is even, and $\{e_1,e_2\}, \ldots, \{e_{l-4},e_{l-3}\},$ $\{e_{l-2},e_{l-1},e_{l}\}$ if $l$ is odd.

For each bin of size $2$ and size $3$, if all the edges are in $G$, we call it a full bin; if none of them are in $G$, we call it an empty bin; otherwise, we call it a half-full bin. Fix a permutation $\pi$, let $s_{\pi}$ be the number of half-full bins. For each half-full bin, it contributes at least one to $\diff(G,\pi(G))$, therefore $s_{\pi}\leq |\diff(G,\pi(G))| = 2 (m - \aut(G,\pi))$. We have
\ifnum\full=1
    \begin{multline*}
     \Pr[\text{$G$ is $(\beta_0, \beta_0/240)$-asymmetric}|\calB] =   \Pr[\forall \pi: \id(\pi) \leq (1-\beta_0)n, \aut(G,\pi)\leq (1-\beta_0/240) m | \calB] \\ \geq \Pr[\forall \pi: \id(\pi) \leq (1-\beta_0)n, s_{\pi} \geq \beta_0  m /120 |\calB].
     \end{multline*}
\fi

\ifnum\full=0
$     \Pr[\text{$G$ is $(\beta_0, \beta_0/240)$-asymmetric}|\calB] =   \Pr[\forall \pi: \id(\pi) \leq (1-\beta_0)n, \aut(G,\pi)\leq (1-\beta_0/240) m | \calB] \\ \geq \Pr[\forall \pi: \id(\pi) \leq (1-\beta_0)n, s_{\pi} \geq \beta_0  m /120 |\calB].$
\fi

Now we turn to show that
\ifnum\full=1
\[ 
\fi
\ifnum\full=0
$
\fi
\Pr[\forall \pi: \id(\pi) \leq (1-\beta_0)n, s_{\pi} \geq \beta_0  m /120 |\calB] \geq 1 - n^{-17}. 
\ifnum\full=1
\] 
\fi
\ifnum\full=0
$
\fi
To show this, we only have to prove for every  $\beta$ where $\beta \geq \beta_0$,
\begin{gather}
 \Pr[\forall \pi: \id(\pi) = (1-\beta)n, s_{\pi} \geq \beta  m /120 |\calB] \geq 1 - n^{-18}, \label{eqn:conditionedimproved-1}
\end{gather}
and take a union bound over (at most $n$ possible) $\beta$'s.

Fix $\beta : \beta \geq \beta_0$ and fix a permutation $\pi$ such that $\id(\pi) = (1 - \beta)n$. Let $\calC_t$ be the event that $\calB$ happens and there are $t$ edges in $G$ fall into the bins of size $2$ and $3$. Since $\calB$ is a disjoint union of $\calC_t$ for all $t : \beta m /60 \leq t \leq 2\beta Z/10^9$, we have
\ifnum\full=1
\[
\fi
\ifnum\full=0
$
\fi
\Pr[s_{\pi} \leq \beta  m /120 |\calB] = \sum_{t = \beta m/60}^{m} \Pr[s_{\pi} \leq \beta  m /120 |\calC_t] \cdot \Pr[\calC_t | \calB].
\ifnum\full=1
\]
\fi
\ifnum\full=0
$
\fi
We will prove that
\begin{gather}
\Pr[s_{\pi} \leq \beta  m /120 |\calC_t] = \left( {n \choose \beta n} (\beta n)!\right)^{-1} \cdot n^{-18},
\label{eqn:conditionedimproved-2}
\end{gather}
and by taking a union bound over all ${n \choose \beta n} (\beta n)!$ possible $\pi$'s, we prove \eqref{eqn:conditionedimproved-1}.

Let $\gamma = \beta / 120$. Let $L$ be the number of possible edges in bins of size $2$ and $3$. We have $L= Z - {(1 - \beta) n \choose 2} - \frac{\beta n}{2} \geq \beta n^2/4$ (for large enough $n$). Together with $t \leq 2\beta Z /10^9$, we have $t \leq 4L/10^9 \leq L/10^8$. Let $B$ be the number of bins of size $2$ and $3$. Fix $t$ such that  $\beta m / 60 \leq t \leq L/10^8$. Conditioned on $\calC_t$, the $\binom{L}{t}$ ways to select these $t$ edges are uniformly distributed. Now we compute the number of ways such that there are at most $2\gamma m$ half-full bins. Suppose that there are $i$ half-full bins (for $i \leq 2\gamma m \leq t/2$). There are ${B \choose i}$ ways to choose these bins. There are at most ${B \choose (t-i)/2}$ ways to choose the full bins (since $t/2 < L/2 \cdot 10^8 < B/2$). For each half-full bin, there are at most $6$ ways to choose the edges in the bin. Therefore,
\begin{gather}
\Pr[s_\pi = i |\calC_t] \leq  6^i {B \choose i} {B \choose \frac{t-i}{2}} {L \choose t}^{-1} \leq 6^i \cdot \frac{L^{\frac{t+i}{2}}}{i! \left(\frac{t-i}{2}\right)!} \left(\frac{t}{L}\right)^t = 6^i L^{-\frac{t-i}{2}} \frac{t^t}{i! \left(\frac{t-i}{2}\right)!}. \label{eqn:conditionedimproved-5}
\end{gather}
Since $i! \geq (i/e)^i$, we have
\ifnum\full=1
\begin{gather}
\eqref{eqn:conditionedimproved-5} \leq 6^i  L^{-\frac{t-i}{2}} e^{\frac{t+i}{2}} \frac{t^t}{i^i \left(\frac{t-i}{2}\right)^{\frac{t-i}{2}}} \leq  6^i  L^{-\frac{t-i}{2}} (2e)^{\frac{t+i}{2}} \frac{t^t}{i^i (t-i)^{\frac{t-i}{2}}} .\label{eqn:conditionedimproved-6}
\end{gather}
\fi
\ifnum\full=0
\vspace{-6ex}
\begin{gather}
\qquad\qquad\qquad\qquad\qquad\eqref{eqn:conditionedimproved-5} \leq 6^i  L^{-\frac{t-i}{2}} e^{\frac{t+i}{2}} \frac{t^t}{i^i \left(\frac{t-i}{2}\right)^{\frac{t-i}{2}}} \leq  6^i  L^{-\frac{t-i}{2}} (2e)^{\frac{t+i}{2}} \frac{t^t}{i^i (t-i)^{\frac{t-i}{2}}} .\label{eqn:conditionedimproved-6}
\end{gather}
\fi
Using $i^i (t-i)^{\frac{t-i}{2}} \geq \left(\frac{t}{4}\right)^{\frac{t+i}{2}}$, we have
\ifnum\full=1
\begin{gather}
\eqref{eqn:conditionedimproved-6} \leq 6^i  (8e)^{\frac{t+i}{2}}  \left(\frac{t}{L}\right)^{\frac{t-i}{2}} \leq (48e)^{\frac{3t}{4}}\left(\frac{t}{L}\right)^{\frac{t}{4}} \leq \left(\frac{10^7t}{L}\right)^{\frac{t}{4}}.
\label{eqn:conditionedimproved-7}
\end{gather}
\fi
\ifnum\full=0
\vspace{-6ex}
\begin{gather}
\qquad\qquad\qquad\qquad\qquad\qquad\qquad\eqref{eqn:conditionedimproved-6} \leq 6^i  (8e)^{\frac{t+i}{2}}  \left(\frac{t}{L}\right)^{\frac{t-i}{2}} \leq (48e)^{\frac{3t}{4}}\left(\frac{t}{L}\right)^{\frac{t}{4}} \leq \left(\frac{10^7t}{L}\right)^{\frac{t}{4}}.
\label{eqn:conditionedimproved-7}
\end{gather}
\fi
Since $\beta m / 60 \leq t \leq L/10^8$, and $(10^7t/L)^{t/4}$ is monotonically decreasing when $t < L/(10^7e)$, we have
\ifnum\full=1
\[
\fi
\ifnum\full=0
$
\fi
\eqref{eqn:conditionedimproved-7} \leq \left(\frac{10^7 \beta m}{60L}\right)^{\frac{\beta m}{240}} \leq \left(\frac{10^7\beta c n}{60 \beta n^2 / 4}\right)^{\frac{\beta cn}{240}} \leq \left(\frac{10^6c}{n}\right)^{\frac{\beta cn}{240}} \leq \left(\frac{10^{10}}{n}\right)^{40 \beta n } \leq \left(\frac{1}{n}\right)^{30\beta n},
\ifnum\full=1
\]
\fi
\ifnum\full=0
$
\fi
where the second last inequality is because $10^4 \leq c \leq n/10^7$ and $(10^6c/n)^{\beta cn/240}$ is monotonically decreasing in this range, and the last inequality is for large enough $n$. Observing that ${n \choose \beta n} (\beta n)! \leq n^{2\beta n}$ and $\beta \geq \beta^*\geq 1/n$, we proved that
\ifnum\full=1
\[
\fi
\ifnum\full=0
$
\fi
\Pr[s_\pi = i |\calC_t] \leq    \left( {n \choose \beta n} (\beta n)!\right)^{-1} \cdot n^{-20} .
\ifnum\full=1
\]
\fi
\ifnum\full=0
$
\fi
By taking a union bound over all (at most $n^2$ many) $i$'s such that $i \leq 2\gamma m$, we prove \eqref{eqn:conditionedimproved-2}.
\end{proof}

\pref{thm:robust-asymmetry} is proved by combining \pref{cor:balanced-edges}, \pref{lem:conditionedimproved}, and taking a union bound over all possible $\beta = \beta_0$ (where there are at most $n$ of them).

\ifnum\full=1
\subsection{Generalization to hypergraphs} \label{sec:robust-asymmetry-hypergraph}

In this subsection, we generalize \pref{thm:robust-asymmetry} to random $k$-uniform hypergraphs for any constant $k \geq 3$.
\fi

\ifnum\full=0
\paragraph{Generalization to hypergraphs.} We generalize \pref{thm:robust-asymmetry} to random $k$-uniform hypergraphs for any constant $k \geq 3$.  The proof is deferred to the full version due to space constraints. 
\fi


\begin{theorem} \label{thm:randomhypergraphwhp}
   For any constant $k$, there exists constant $\kappa_k$, such that for $m=cn$ where $\kappa_k \leq c \leq \binom{n}{k}/\kappa_k^3$, $n$ large enough, if we set $\beta^* = \max\{\exp(-c/6),1/n\}$, with probability $(1 - n^{-15})$, for all $\beta : \beta^* \leq \beta \leq 1$, a random graph $H$ from the distribution $\calG^{(k)}_{n,m}$ is $(\beta,\beta/240)$-asymmetric. For $k=3$, $\kappa_3 = 10^{4}$ suffices.
\end{theorem}

\ifnum\full=1
The proof of \pref{thm:randomhypergraphwhp} mostly follows the lines of the proof of \pref{thm:robust-asymmetry}. But we need some small modifications. For simplicity, we only prove the theorem  for $k=3$. For higher $k$, we encourage the readers to check by themselves. 

Now we work with $k = 3$. For any permutation $\pi$ over the vertex set $V$, we define a directed graph $G^{(3)}_{\pi}=(\binom{V}{3},E_{\pi})$, and $(e_1,e_2)\in E_{\pi}$ if and only if $e_2=\pi(e_1)$. Since each  $e \in {V \choose 3}$ has in-degree and out-degree exactly 1, we can divide $G^{(3)}_{\pi}$ into disjoint unions of directed cycles. Similarly as in the ordinary graph case, we call each directed cycle a bin, and the size of the bin is the number of elements in the cycle. Let $Z_3=\binom{n}{3}$.

\begin{fact}
  For any size-1 bins, there are only three situations:
    \begin{itemize}
        \item $e=\{u,v,w\}$ where $u$, $v$  and $w$ are all fixed points of $\pi$. We call these bins \emph{type-1} size-1 bins.The number of type-1 size-1 bins is at most $\binom{\id(\pi)}{3}$.
        \item $e=\{u,v,w\}$ where one of them is a fixed point of $\pi$ and the other two map to each other under $\pi$. We call these classes type-2 size-1 bins. The number of type-2 size-1 bins is at most $\id(\pi)\cdot \frac{n-\id(\pi)}{2}=O(n^2)$ for any permutation $\pi$.
        \item $e=\{u,v,w\}$ where $\pi(u)=v$ and $\pi(v)=w$ and $\pi(w)=u$. We call these classes type-3 size-1 bins. The number of type-3 size-1 bins is at most $\frac{n-\id(\pi)}{3}$ for any permutation $\pi$.
    \end{itemize}
\end{fact}

The following lemma is an analogue of \pref{lem:balancededgesimproved}.

\begin{lemma} \label{lem:balancedhyperedgesimproved}
    For any fixed $\beta$ where $\beta\geq \beta^*$, let $\mathcal{D}'$ be the event that for any permutation $\pi$ such that $\id(\pi)= (1-\beta) n$, the number of the  hyperedges in $H$ that fall into the bins of size $\geq 2$ is at least $\beta m /60$ and at most $2\beta Z_3/10^9$. We have
    \[
        \Pr_{H\sim \mathcal{H}^{(3)}_{n,m}}[\mathcal{D}'] = 1- n^{-\omega(1)}.
    \]
\end{lemma}

By a union bound over all the $\beta : \beta \geq \beta^*$ (there are at most $n$ of them), we get the following corollary.
\begin{corollary} \label{cor:hyperbalanced-edges}
 Let $\mathcal{D}$ be the event that for every $\beta \geq \beta^*$, any permutation $\pi$ such that $\id(\pi)= (1-\beta) n$, the number of the  hyperedges in $H$ that fall into the bins of size $\geq 2$ is at least $\beta m /60$ at at most $2\beta Z_3/10^9$. We have:
    \[
        \Pr_{H\sim \mathcal{H}^{(3)}_{n,m}}[\mathcal{D}] = 1- n^{-\omega(1)}.
    \]
\end{corollary}

The proof of \pref{lem:balancedhyperedgesimproved} is similar to that of \pref{lem:balancededgesimproved}, except that now we also need to take care of type-2 size-1 bins. 

\begin{lemma} \label{lem:hypertype2}
    With probability $1-n^{\omega(1)}$, for any permutation $\pi$ with $\id(\pi)\leq(1-\beta^*)n$,  the number of selected hyperedges in type-2 size-1 bins is at most $c\beta n/1000$.
\end{lemma}

\begin{proof}
We first prove that for every $\beta \geq \beta^*$, with probability $1-n^{-\omega(1)}$, for every permutation $\pi$ with $\id(\pi)=(1-\beta) n$, the number of selected hyperedges in type-2 size-1 bins is at most $c\beta n/1000$. By a union bound over all possible $\beta$'s (where there are at most $n$ of them), we get the desired statement.

For any fixed permutation $\pi$ with $\id(\pi)=(1-\beta) n$, the number of type-2 size-1 bins is at most $(1-\beta) n \beta n/2\leq \beta n^2/2$. For each possible hyperedge $e$, we define the random variable $X_{e}$ as the indicator variable for the event that $e$ is selected as an hyperedge in $H$. Note that $\E[X_e]=cn/Z_3 \leq \frac{7c}{n^2}$. Define random variable $X=\sum_{e\text{~in type-$2$ size-$1$ bin}}X_e$ as the number of selected hyperedges in type-2 size-1 bins, by linearity of expectation,  
    \[
        \E[X]\leq \frac{\beta n^2}{2}\cdot \frac{7c}{n^2}<4c\beta.
    \]
     On the other hand, we can also show that all these random variables are negative associated, therefore through Chernoff bound for negative associated random variables, we have
    \[
        \Pr_{H\sim \calG^{(3)}_{n,m}}[X \geq \beta cn/1000]\leq \exp(-1/3\cdot (n/250-1)^2 \cdot 4c\cdot \beta)\leq \exp(-c\beta n^2/10^5)
    \]

    By a union bound over at most $\binom{n}{\beta n}(\beta n)!\leq n^{2\beta n}$ such permutations, the probability that there exists $\pi$ with $\id(\pi) = (1 - \beta)n$ such that the number of type-$2$ size-$1$ bins is more than $c \beta n$, is at most
    \[
        \exp(-c\beta n^2/10^5)\cdot n^{2\beta n}= \exp(-c\beta n^2/(10^5)+2\beta n\log n)\leq \exp(-c\beta n^2/(10^6)) \leq n^{-\omega(1)}.
    \]

\end{proof}

The following lemma is an analogue of \pref{lem:incident}, and the proof is almost identical. 
\begin{lemma}  \label{lem:hyperincident}
    Let $H \sim \calG^{(3)}_{n,m}$. Suppose that $\beta \geq \beta^*$. With probability $1-n^{-\omega(1)}$, for any $T\subseteq V$, $|T|= \beta n$, the number of hyperedges incident to $T$ is at least least $c \beta n / 40$.
\end{lemma}

\begin{proof}[Proof of \pref{lem:balancedhyperedgesimproved}]
We only establish the lower bound ($\beta m/60$). The proof for upper bound is almost identical to that in \pref{lem:balancededgesimproved}.

Let $T$ be the set of non-fixed points of $\pi$, then $|T|=\beta n$. By \pref{lem:hyperincident}, we know that with probability $(1 - n^{-\omega(1)})$, there are at least $c\beta n/40$ hyperedges incident to $T$ -- all these edges are either in bins of size $\geq 2$, or edges in size-1 bins of type-2 or type-3. By \pref{lem:hypertype2}, we know that with probability $1-n^{\omega(1)}$ the number of selected hyperedges that fall into type-2 size-1 bins is at most $c\beta n/1000$. Finally we recall that there are at most $\frac{n - \id(\pi)}{3} = \beta n / 3$ type-3 size-1 bins. 

Therefore, with probability $(1 - n^{-\omega(1)}$, the number of selected hyperedges that fall into bins of size $\geq 2$ is at least $c\beta n/40 - \beta n/3-c\beta n/1000 \geq c\beta n/60$.
\end{proof}

Finally we state the following analogue of \pref{lem:conditionedimproved} (whose proof is also almost identical). 
 \begin{lemma}       \label{lem:hyperconditionedimproved}
Conditioned on event $\mathcal{D}$,  for $10^4 \leq c \leq n/10^7$, $\beta \geq \beta^*$, with probability $(1-n^{-17})$, $H\sim \calG^{(3)}_{n,m}$ is $(\beta,\beta/240)$-asymmetric.
\end{lemma}

The $k=3$ case in \pref{thm:randomhypergraphwhp} follows from \pref{cor:hyperbalanced-edges} and \pref{lem:hyperconditionedimproved}.

\fi
\ifnum\full=1
\newcommand{\ABEL}{\mathsf{Additive}\mathrm{-}\mathsf{CSP}(\pred)}
\newcommand{\pred}{\psi}

\section{Generalizing Feige's Hypothesis}\label{sec:feige}
\rnote{Toran seems to claim all these results after his Lemma 3.2.  He does it explicitly for $\Z_m$.}

In this section, we show how our main results can be derived from a much broader class of assumptions than just Feige's \RXOR hypothesis.  In doing so, we are inspired by the recent $\NP$-hardness results of~\cite{Cha13}, along with the Lasserre gaps of~\cite{Tul09}.  The focus of these papers was on a specific type of predicates, which we now define.

\begin{definition}
Let $k \geq 3$, and let $H$ be a finite abelian group.  Then a predicate $\pred:H^k \rightarrow \{0, 1\}$ is a \emph{balanced pairwise independent subgroup} of $H^k$ if the following two conditions hold.
\begin{itemize}
\item For a uniformly random element $a \sim \pred$, each coordinate $a_i$ is uniformly distributed over $G$.  Furthermore, each pair of coordinates $a_i$ and $a_j$, for $i \neq j$, is independent.
\item $\pred$ is a subgroup of $H^k$.
\end{itemize}
Here we are using the convention that $\pred$ is both a function on $H^k$ and a subset of $H^k$, i.e. the subset $\pred^{-1}(1)$.  We will everywhere assume that $\pred$ is a proper subgroup of $H^k$, as otherwise $\pred$ is uninteresting. 
\end{definition}

An instance $\calC$ of the $\ABEL$ problem consists of a set of constraints $\calC = (C_1, \ldots, C_m)$, each of the form
\begin{equation*}
\pred(x_{j_1} + a_1, \ldots, x_{j_k} + a_k) = 1,
\end{equation*}
where the $a_i$'s are elements of $H$.  The corresponding homogeneous instance $\homog{\calC} = \{\homog{C_1}, \ldots, \homog{C_m}\}$ is formed by setting all of the $a_i$'s in each constraint to $0$.  We note that because $\pred$ is a subgroup of $H^k$, it contains $0^k$, so the assignment $x_1 = \ldots = x_n = 0$ satisfies every constraint in $\homog{\calC}$.

For a fixed number of variables $n$ and constraints $m$, a random instance of $\ABEL$ is formed by generating $m$ constraints of this form independently and uniformly.  This involves, for each constraint, picking the variables $x_{j_1}, \ldots, x_{j_k}$ uniformly from the $\binom{n}{k} \cdot k!$ possibilities, along with choosing each $a_i$ independently and uniformly at random.  The next hypothesis states that random $\ABEL$ instances are indistinguishable from almost-satisfiable instances.

\paragraph{Random $\ABEL$ Hypothesis.}  \emph{For every fixed $\eps > 0$, $\Delta \in Z^+$, there is no polynomial time algorithm which on almost all $\ABEL$ instances with $n$ variables and $m = \Delta n$ constraints outputs ``typical'', but which never outputs ``typical'' on instances which an assignment satisfying at least $(1-\eps)m$ constraints.}
\newline

Note that Feige's \RXOR hypothesis can be recovered by setting $\pred = \mathrm{\threexor}$.  Using this, we can weaken the assumptions for \pref{thm:r3xor-hardness} so that it only needs the Random $\ABEL$ Hypothesis to be true for some balanced pairwise independent subgroup $\pred$.

\begin{theorem}\label{thm:abel}
Let $\pred$ be a predicate which is a balanced pairwise independent subgroup.  Assume the Random $\ABEL$ Hypothesis.  Then there is no polynomial-time algorithm for \robustgiso.  More precisely, there is an absolute constant $\delta >0$ such that the following holds: suppose there is a $t(n)$-time algorithm which can distinguish $(1-\eps)$-isomorphic $n$-vertex, $m$-edge graph pairs from pairs which are not even $(1-\delta)$-isomorphic (where $m = O(n)$).  Then there is a  universal constant $\Delta \in \Z^+$ and a $t(O(n))$-time algorithm which outputs ``typical'' for almost all $n$-variable, $\Delta n$-constraint instances of the $\ABEL$ problem, yet which never outputs ``typical'' on instances which are $(1-\eps)$-satisfiable.
\end{theorem}

\paragraph{Outline of proof.}
The proof of Theorem~\pref{thm:abel} largely follows the same outline as the proof of \pref{thm:r3xor-hardness}.  As a result, our proof of \pref{thm:abel} will be carried out at a somewhat high level, except in those areas where we need to highlight changes.  The most significant change is that the general abelian group setting requires us to develop more complicated gadgets, and we will need to describe these gadgets before stating the reduction.

\subsection{Abelian variable gadgets}\label{sec:variablegadget}

For a given variable $x$, the construction in \pref{sec:reduction-3xor} creates two vertices $x \mapsto 0$ and $x \mapsto 1$ and places an edge between them.  This subgraph has two automorphisms, corresponding to the two possible Boolean assignments to $x$.  This is a simple example of what we will call a \emph{variable gadget}.  When $x$ takes values from a more general group $H$, a variable gadget for $x$ will be a graph $G = (V_x \cup V_A, E)$ .  Here, $V_x$ is the set of vertices $``x \mapsto a"$ for $a \in H$, and the vertices in $V_A$ are thought of as auxiliary vertices.  Furthermore, we require that the only automorphisms of $G$ correspond to the $|H|$ possible assignments to $x$.  More formally:

\begin{definition}
Let $x$ be a variable which takes values from $H$, an abelian group.
A \emph{variable gadget} for $x$ is a graph $G = (V_x \cup V_A, E)$ satisfying
\begin{itemize}
\item[] \textbf{(Completeness)} For each $a \in H$, $G$ has an automorphism $f_a$ such that $f_a(x\mapsto b) = x\mapsto(a+ b)$ for each $b \in H$.
\item[] \textbf{(Soundness)} Let $f:V_x \cup V_A \rightarrow V_x \cup V_A$ be an automorphism on $G$.  Then there exists an element $a \in H$ such that $f(x\mapsto b) = x\mapsto(a+ b)$ for each $b \in H$.
\end{itemize}
\end{definition}

Constructing a variable gadget for $x$ when $x$ is non-Boolean is more complicated than just slapping an edge between two vertices.  However, we are able to construct variable gadgets with the following parameters:

\begin{lemma}\label{lem:variablegadget}
Let $x$ be a variable over $H$, an abelian group.  Then there is a variable gadget for $x$ with at most ${4\cdot|H| \log^2 |H|}$ auxiliary variables, $7\cdot |H| \log^2 |H|$ edges, and a maximum clique size of $\min\{4, |H|\}$.
\end{lemma}
The proof of this lemma can be found in \pref{sec:gadgetproofs}.

\subsection{Label-extended graphs}\label{sec:labelext}
Let $\pred:H^k \rightarrow \{0, 1\}$ be a subgroup, and let $C$ be an $\ABEL$ constraint.
The \emph{label-extended graph} of $C$ consists of vertex sets $V_{x_i}$ as defined in \pref{sec:variablegadget} for each $i \in[k]$; and $|\pred|$~``constraint vertices'' with names $``x_1\mapsto a_1, \ldots, x_k \mapsto a_k"$ for each assignment $(a_1, \ldots, a_k)$ which satisfies $C$.  We will sometimes write these constraint vertices as assignments $\alpha:\{x_1, \ldots, x_k\} \rightarrow H$, and the set of them we call $V_C$.  For edges, each constraint vertex $\alpha$ is connected to each of the $k$ variable vertices it is consistent with.

Let $G$ be the label-extended graph of a constraint $C$, and let $\homog{G}$ be the label-extended graph of the homogeneous constraint $\homog{C}$.  Our goal is to show that every homomorphism from $G$ to $\homog{G}$ corresponds to a satisfying assignment of $C$.  While this may not be true in general, we will be able to restrict our attention to only those automorphisms which act on each set $V_{x_i}$ in a very limited way; this is due to the variable gadget from \pref{sec:variablegadget}.  Let $\alpha:\{x_1, \ldots, x_k\} \rightarrow H$ be an assignment to the variables.  Then a permutation $f$ on the vertices of $G$ is an \emph{$\alpha$-assignment} if for each $i \in [k]$, $f(x_i \mapsto b) = x_i \mapsto(b- \alpha(x_i) )$ for each $b \in H$.  Then

\begin{lemma}\label{lem:constraintvariable}
Let $\pred:H^k \rightarrow \{0, 1\}$,  let $C$ be an $\ABEL$ constraint, and let $G$ be its label-extended graph. Let $\homog{G}$ be the label extended graph of $\homog{C}$. Then
\begin{itemize}
\item[] \textbf{(Completeness)} For each assignment $\alpha$ which satisfies $C$, there is an $\alpha$-permutation $f$ which is a homomorphism from $G$ to $\homog{G}$.
\item[] \textbf{(Soundness)} Let $f$ be an $\alpha$-permutation which is a homomorphism from $G$ to $\homog{G}$.  Then $\alpha$ is a satisfying assignment of $C$.
\end{itemize}
\end{lemma}
In fact, though we will not use this, it can be shown that the soundness holds even if $f$ is only a $(1-\frac{1}{k}+\eps)$ homomorphism from $G$ to $\homog{G}$, for any $\eps > 0$.  \pref{lem:constraintvariable} is proved in \pref{sec:labelext}.

\subsection{Reduction from $\ABEL$ to \giso}

We begin with a generalization of \pref{def:gadgetgraph}:
\begin{definition}
Let $C$ be a $\pred$ constraint involving variables $x_1, \ldots, x_k$.  The associated gadget graph $G_C$ has a vertex set consisting of the variable vertices $V_{x_i}$ for each $i \in [k]$ and the constraint vertices $V_C$.  For each $i \in [k]$, the variable gadget from \pref{lem:variablegadget} is added to $V_{x_i}$; this adds a set of auxiliary variables, which we call $V_{A_i}$.   Next, $V_C$ is connected to the variable vertex sets via the label-extended graph.  Finally, the variables in $V_C$ are connected to each other by a clique.
\end{definition}

Here the vertex sets $V_{x_i}$ are as defined in \pref{sec:variablegadget} and the vertex set $V_C$ is as defined in \pref{sec:labelext}. 

Let $\calC$ be a collection of $\ABEL$ constraints.  We define the associated graph $G_{\calC}$ as in \pref{def:3xorencoding} as a union of gadget graphs $G_C$ where the variable vertex sets and gadgets are associated with each other.

The number variable vertices in $G_{\calC}$ is $|H| \cdot n$, the number of constraint vertices is $|\pred|\cdot m$, and the number of auxiliary variables given in \pref{lem:variablegadget} is at most $4\cdot|H| \log^2 |H| \cdot n$.  Next, the number of edges from the complete graphs and the constraint-variable graphs is $\left(\binom{|\pred|}{2} + |\pred|\cdot k\right)$, and the number of edges given by \pref{lem:variablegadget} is at most $7\cdot |H| \log^2 |H| n$.  This gives us the following remark:

\begin{remark}
If $\calC$ is an $\ABEL$ instance with $n$ vertices and $m$ constraints then the graph $G_{\calC}$ has $N = N_1 m + N_2 n$ vertices and $M = M_1 m + M_2 n$ edges. Here the $N_i$'s and $M_i$'s are absolute constants which depend only on the group $H$ and the predicate $\phi$.
\end{remark}

One thing that we will need is that the largest cliques of $G_{\calC}$ are those only involving constraint vertices.  To show this, we will first need the following simple proposition:
\begin{proposition}\label{prop:bigsubgroup}
Let $H$ be a finite abelian group and let $\pred:H^k \rightarrow \{0, 1\}$ be a balanced pairwise independent subgroup.  Then $|\pred| \geq |H|^2$.
\end{proposition}
\begin{proof}
For a uniformly random element $a \sim \pred$, the definition of a balanced pairwise independent subgroup states that $a_1$ and $a_2$ will each be uniform on $H$ and independent of each other.  Thus, $(a_1, a_2)$ ranges over the $|H|^2$ possibilities with equal probability, meaning that $\pred$ must have at least $|H|^2$ elements.
\end{proof}
\begin{remark}\label{rem:largestcliques}
By \pref{lem:variablegadget}, the largest clique in the variable gadgets is of size $\min\{4, |H|\}$, whereas the constraint vertices come in cliques of size $|\pred| \geq |H|^2$.  Thus, the largest cliques in the graph consist of the constraint vertices.
\end{remark}

\paragraph{The reduction.}  \emph{Given a collection of $\ABEL$ constraints $\calC$, we introduce the corresponding \giso instance $(G_{\calC}, G_{\homog{\calC}})$.}

\subsection{Abelian proofs}

Having defined the reduction from $\ABEL$ to \giso, we are now in a position to prove \pref{thm:abel}.  Our two main lemmas are:

\begin{lemma}\label{lem:chancompleteness}
Let $\calC$ be a $\ABEL$ instance such that $\val(\calC) \geq 1-\eps$.  Then $\GI(G_{\calC}, G_{\homog{\calC}}) \geq 1- \eps$.
\end{lemma}

\begin{lemma}\label{lem:chansoundness}
There are absolute constants $\delta, c > 0$ depending only on the predicate $\pred$ such that the following holds: let $\calC = \{C_1, C_2, \ldots, C_m\}$ be a random $\ABEL$ instance with $n$ variables and $m = cn$ equations.  With probability $1-o(1)$, we have
\begin{equation*}
\GI(G_{\calC}, G_{\homog{\calC}}) < 1- \delta.
\end{equation*}
\end{lemma}

Given a polynomial-time algorithm $\calA$ for \robustgiso, we can define a polynomial-time random $\ABEL$ solver $\calA_{\pred}$ which, on input $\calC$, runs $\calA$ on $\left(G_{\calC}, G_{\homog{\calC}}\right)$.  If $\val(\calC) \geq 1-\eps$, then by \pref{lem:chancompleteness}, $\GI(G_{\calC}, G_{\homog{\calC}}) \geq 1-\eps$, and so $\calA_{\pred}$ will never output ``typical".  On the other hand, if $\calC$ is a random $\ABEL$ instance, then with probability $1-o(1)$ we have $\GI(G_{\calC}, G_{\homog{\calC}}) < 1-\delta$  by \pref{lem:chansoundness}, so $\calA_{\pred}$ will almost always output ``typical''.  However, this contradicts the Random $\ABEL$ hypothesis, meaning that no such \robustgiso algorithm can exist.

We now proceed to the proofs.

\paragraph{Completeness}

This case follows directly from the completeness properties of our gadgets.
\begin{proof}[Proof of \pref{lem:chancompleteness}]
Let $\tau$ be an assignment to the variables in $\calC$ such that $\val(\calC;\tau) \geq 1-\eps$.  Now we define a bijection $\pi$ from the variables in $G_{\calC}$ to the ones in $G_{\homog{\calC}}$.

For each variable vertex $x_j$, set $a_j := \tau(x_j)$.  By \pref{lem:variablegadget} there is an automorphism $f_{a_j}:(V_{x_j} \cup A_{x_j}) \rightarrow (V_{x_j} \cup A_{x_j})$ such that $f(x_j \rightarrow b) = x_j \rightarrow (b-a_j)$ for each $b \in H$.  Set $\pi$ to be consistent with $f_{a_j}$ on $V_{x_j} \cup A_{x_j}$ for each $j \in [n]$.  

Next, let $C_i$ be a constraint which is satisfied by $\tau$, and assume without loss of generality that it is applied to the variables $x_1, \ldots, x_k$.  Consider the subgraph $G_i$ which is the induced subgraph of $G_{\calC}$ restricted to the vertices in $V_{x_1}, \ldots, V_{x_k}$ and $V_{C_i}$.  Define $\homog{G_i}$ to be the analogous subgraph of $G_{\homog{C}}$.  Then 
 $\pi$ restricted $G_i$ is an $(a_1, \ldots, a_k)$-permutation, as defined in \pref{sec:labelext}, and so \pref{lem:constraintvariable} says that it can be extended to a homomorphism between $G_i$ and $\homog{G_i}$.  Define $\pi$ on $V_{C_i}$ accordingly.

This leaves how to define $\pi$ on the variables in $V_{C_i}$ when $\tau$ does not satisfy $C_i$, and here we allow $\pi$ to map $V_{C_i}$ arbitrarily into $V_{\homog{C_i}}$. Until this step, $\pi$ was a homomorphism on every subgraph it was defined on.  Thus, the only edges violated by $\pi$ are those involved in the label-extended graph of any constraint $C_i$ which $\tau$ does not satisfy.  As these constraints make up only an $\eps$-fraction of the total constraints, $\GI(G_{\calC}, G_{\homog{\calC}}) \geq (1-\eps)$. 
\end{proof}

\paragraph{Soundness}

In this section, we sketch the proof of the soundness lemma.  To start with:

\begin{claim}\label{claim:degreebound}
Suppose $c \geq 1$.  A random $k$-uniform hypergraph $G$ drawn from $\calG^{(k)}_{n, m}$, where $m = cn$, is $(1/kc, 100kc)$-degree bounded with probability $1-o(1)$.
\end{claim}

\pref{lem:chansoundness} is directly implied by the following two lemmas.
\begin{lemma}\label{lem:chandecode}
Let $G = ([n], E=\{e_i\})$ be a $k$-uniform hypergraph with $n$ vertices and $m = cn$ hyperedges.  Suppose $H$ is $(\eps, 100kc)$-degree bounded and $(\beta, \gamma)$-asymmetric, where $\gamma \geq 200k\eps$.  Let $\calC = \{C_1, C_2, \ldots, C_m\}$ be an arbitrary $\ABEL$ instance with $n$ variables and $m$ constraints based on $H$.  If we set
\begin{equation*}
\delta := \delta(c, \eps, \beta, \gamma) = \min \left\{\frac{1}{10(M_1 + M_2)}, \frac{\gamma}{4(M_1 + M_2)},
\frac{\eps}{3(M_1 + M_2)N_1 c}\right\},
\end{equation*}
when $\GI(G_{\calC}, G_{\homog{\calC}}) \geq 1-\delta$, we have $\val(\calC) \geq .9 - 100(\eps + \beta)$.
\end{lemma}
\begin{lemma}\label{lem:lowvalue}
Given a predicate $\pred$, there is a constant $c > 0$ such that the following holds:
Let $\calC$ be a random $\ABEL$ instance with $n$ and $m \geq c n$ equations.  With probability $1-o(1)$, we have $\val(\calC) < 0.51$.
\end{lemma}
\begin{proof}[Proof of \pref{lem:chansoundness}]
We will choose a value for $c\geq 10^{10}$ large enough so that our applications of \pref{lem:lowvalue} and \pref{thm:randomhypergraphwhp} work properly.  Set $\eps = \frac{1}{kc}$, $\gamma = \frac{200}{c}$, and $\beta = \frac{48000}{c}$.  Combining \pref{claim:degreebound}, \pref{lem:lowvalue}, and \pref{thm:randomhypergraphwhp}, we know that with probability $(1-o(1))$, all of the following hold:
\begin{enumerate}
\item $G$ is $(\eps, 100kc)$-degree bounded,\label{item:bounded}
\item $G$ is $(\beta, \gamma)$-asymmetric, and\label{item:asymmetric}
\item $\val(\calC) < 0.51$.\label{item:thevalueistoolowman}
\end{enumerate}
Given that these hold, assume for sake of contradiction that $\GI(G_{\calC}, G_{\homog{\calC}}) \geq 1-\delta$, for $\delta$ as defined in \pref{lem:lowvalue}.  Then because $G$ satisfies \pref{item:bounded} and \pref{item:asymmetric}, \pref{lem:lowvalue} implies that
\begin{equation*}
\val(\calC) \geq .9 - 100\left(\frac{1}{kc} + \frac{48000}{c}\right) \geq .8,
\end{equation*}
where the last step follows because $c \geq 10^{10}$.
However, this contradicts \pref{item:thevalueistoolowman}.  Therefore, $\GI(G_{\calC}, G_{\homog{\calC}}) <1-\delta$ with probability $1-o(1)$.
\end{proof}

The proof of \pref{lem:lowvalue} is standard.
\begin{proof}[Proof of \pref{lem:lowvalue}]
First, we note that because $\pred$ is a proper subgroup of $H^k$, $\vert \pred \vert / \vert H^k \vert \leq 1/|H| \leq 1/2$.  Thus, the probability that a random assignment satisfies $\pred$ is at most $1/2$.  Fix an assignment $\alpha$ to the $n$ variables, and let $\calC$ be a random $\ABEL$ instance.  By the Chernoff bound,
\begin{equation*}
\Pr[\text{$\alpha$ satisfies at least $.51m$ constraints}] 
\leq e^{-\Theta(m)} = e^{-\Theta(c\cdot n)}.
\end{equation*}
Union bounding over all the $|H|^n$ possible assignments, the probability that $\val(\calC) \geq .51m$ is at most $|H|^n e^{-\Theta(c\cdot n)}$.  Thus, by taking $c = \Theta(\log |H|)$ this probability is $o(1)$.
\end{proof}

Finally, we give a sketch of the main lemma.

\paragraph{Sketch of \pref{lem:chandecode}.}
Let $\pi$ be a bijection mapping the vertices in $\calL(\calC)$ to the vertices $\calL(\homog{\calC})$ such that $\GI(\calL(\calC), \calL(\homog{\calC}); \pi) \geq 1-\delta$.  Let $A$ be the set of $i \in [m]$ such that $\pi(V_{C_i}) = V_{\homog{C_{i'}}}$ for some $i'$, and let $B$ be the set of $j \in [n]$ such that $\pi(V_{x_j} \cup V_{A_j}) = V_{x_{j'}} \cup V_{A_{j'}}$ for some $j'$.  Then
\begin{claim}
$|A| \geq (1-\delta(M_1+M_2))m$, $|B| \geq (1-2\delta(M_1 + M_2)N_1 c)n$.
\end{claim}
The proof of this statement is essentially identical to the proof of \pref{claim:soundness-correspondence}.  The only difference here is that we use \pref{rem:largestcliques} to show that the unique largest cliques are the constraint vertices. Now, we would like to show that $\pi(V_{x_j})$ acts in a highly structured way for most of the $j \in B$, as in the definition of a variable gadget.  In particular, let $B_{shift}$ be the set of $j \in[n]$ such that for some $j'$ and $a \in H$, $\pi(x_j \mapsto b) = (x_j' \mapsto (b-a))$ for all $b \in H$.  Then
\begin{claim}
$|B_{shift}| \geq (1-3\delta(M_1 + M_2) N_1 c)n$.
\end{claim}
\begin{proof}
Consider $B_{shift}' = B \cap B_{shift}$.  Let $j \in B$ and suppose $x_{j'}$ is the coordinate $\pi$ maps $x_j$ to, i.e. $\pi(V_{x_j} \cup V_{A_j}) = V_{x_{j'}} \cup V_{A_{j'}}$.  If $\pi$ is \emph{not} of the form $\pi(x_j \mapsto b) = x_{j'} \mapsto (b-a)$, then the definition of a variable gadget tells us that $\pi$ cannot be a homomorphism on the induced subgraphs on $V_{x_j} \cup V_{A_j}$ and $V_{x_{j'}} \cup V_{A_{j'}}$.  In other words, $\pi$ must not satisfy one of the edges in $V_{x_j} \cup V_{A_j}$.  Thus, $|B \setminus B_{shift}'| \leq \delta M \leq \delta (\beta_1 + \beta_2) m = \delta (\beta_1 + \beta_2) c n$.  Subtracting this from the size of $B$ gives the claim.
\end{proof}
Now, define a permutation $\sigma$ on the variables in $\calC$ such that $\sigma(j) = j'$ whenever $\pi(V_{x_j}) = V_{x_{j'}}$ and $j \in B_{shift}$.  Furthermore, ensure that $\sigma(j) \neq j$ for any $j \notin B_{shift}$.  Then $\sigma$ is an almost-automorphism for the hypergraph $G = ([n], E)$:
\begin{claim}
$\aut(G;\sigma) \geq (1-100k\eps-2\delta(M_1+M_2))m$.
\end{claim}
By our setting of $\delta$, we have $2\delta(M_1 + M_2) \leq \gamma / 2$, and we also know that $100k\eps \leq \gamma/2$.  Thus, $\aut(G, \sigma)\geq (1-\gamma) m$, and therefore $\sigma$ has at least $(1-\beta)n$ fixed points.

We can now define an assignment $\tau:\{x_j\}\rightarrow H$.  For each $j$ which is not a fixed point of $\sigma$, define $\tau(x_j)$ arbitrarily.  For each $j$ which is a fixed point of $\sigma$, we ensured when constructing $\sigma$ that $j \in B_{shift}$.  Thus, there is some $a \in H$ such that $\pi(x_j \mapsto b) = (x_j \mapsto (b-a))$ for all $b \in H$.  Using this, define $\tau(x_j) = a$.  We will now show the following claim:
\begin{claim}
$\val(\calC, \tau) \geq .9 - 100 (\eps+\beta)$.
\end{claim}
\begin{proof}
Let $E'$ be the set of hyperedges in $E$ whose vertices are all fixed points of $\sigma$, and let $E'' \subset E'$ contain those hyperedges $e_i$ in $E'$ for which all the edges incident to $V_{C_i}$ are satisfied.  By an argument similar to \pref{claim:soundness-threelin-val}, we know that $\vert E''\vert \geq (1 - 100(\eps+\beta) - (M_1 + M_2)\delta)m \geq .9 - 100(\eps+ \beta)$, by our choice of $\delta$.

Given $e_i \in E''$, assume without loss of generality that $e_i = \{x_1, \ldots, x_k\}$.  Then we know that for some $a_1, \ldots, a_k \in H$, $\pi(x_j \mapsto b) = x_j\mapsto (b_j-a_j)$ for each $j$.  Consider the subgraph $G_i$ which is the induced subgraph of $G_{\calC}$ restricted to the vertices in $V_{x_1}, \ldots, V_{x_k}$ and $V_{C_i}$.  Define $\homog{G_i}$ to be the analogous subgraph of $G_{\homog{C}}$.  Then  $\pi$ restricted to $G_i$ is an $(a_1, \ldots, a_k)$-permutation, and it is a homomorphism from $G_i$ to $\homog{G_i}$.  Therefore, \pref{lem:constraintvariable} tells us that $(a_1, \ldots, a_k)$ is a satisfying assignment of $C_i$.  Thus, for each $T_i \in T''$, the assignment $\tau$ satisfies the constraint $C_i$.  As $T''$ makes up at least a $.9 - 100(\eps+\beta)$ fraction of all the constraints, the claim follows.
\end{proof}

\subsection{Variable gadget proof}\label{sec:gadgetproofs}
\begin{proof}[Proof of \pref{lem:variablegadget}.]

Because $H$ is abelian, it is isomorphic to the group $\Z_{p_1}\oplus \Z_{p_2}\oplus \cdots \oplus \Z_{p_t}$, for some integers $p_1, \ldots, p_t$.  Thus, we will without loss of generality assume that it equals this group.  The group $H$ can be visualized as being a $p_1$-by-$p_2$-by-\ldots-by-$p_t$ grid.  For example, the case when $H = \Z_3 \oplus \Z_5$ is depicted in \pref{fig:z3z5}. 

 The important subsets of $H$ for our gadget construction will be those in which a specified coordinate is allowed to be free, and the rest of the coordinates are fixed.  We will call these subsets \emph{rows}.  For example, a $1$-row of $H$ consists of the elements
\begin{equation*}
(0, a_2, \ldots, a_t), (1, a_2, \ldots, a_t), \ldots, (p_t - 1, a_2, \ldots, a_t),
\end{equation*}
where $a_i$ is an element of $\Z_{p_i}$ for each $i$.  For a coordinate $i \in [t]$, define an $i$-row analogously.  In \pref{fig:z3z5}, the rows of the grid correspond to $1$-rows, and the columns correspond to $2$-rows.

\begin{figure}
\centering
\includegraphics{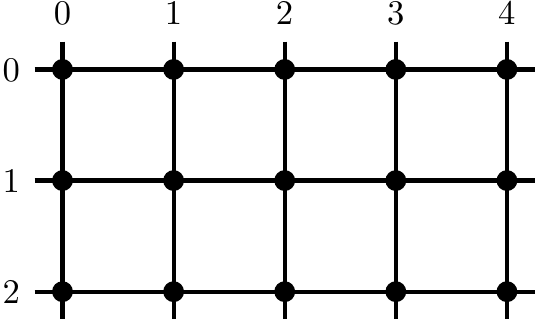}
\caption{A grid representation of the group $\Z_3 \oplus \Z_5$.  Each intersection point corresponds to a group element.}
\label{fig:z3z5}
\end{figure}

We now construct the gadget in two steps.  Each step will involve adding a set of vertices and edges which is local to each $i$-row.

\paragraph{The row gadget:}
Consider the group $H = \Z_3 \oplus \Z_3$.  In this case, the function $f:H \rightarrow H$ in which $f(a, b) = (b, a)$ is an automorphism on $H$.  However, it is not of the form required by the soundness condition; one way to show this is that $f(1, 0) - f(0, 0) = (0, 1)$, whereas the soundness requires that it equal $(1, 0)$.  The row gadget will disallow such automorphisms by forcing any automorphism to maintain the order of coordinates.

We define the gadget $G_{row}$ as follows: the graph initially consists of a vertex $``b"$ for each group element $b \in H$.  These are the variable vertices, the set of which we call $V_x$. (In the rest of the paper, we are naming these vertices $``x \mapsto b''$.  However, we are opting for the simpler option for this gadget construction.)  Note that $V_x = H$, and so $i$-rows of $H$ are $i$-rows of $V_x$, and vice versa. Now, if $t = 1$, then we stop here.  Otherwise, for each coordinate $i$, let $R$ be an $i$-row in $V_x$.  Create a star graph with $i$ degree-one outer nodes and one degree-$i$ inner node, and attach this inner node to every vertex in $R$.  Doing this for each $i$-row completes the construction of $G_{row}$.  Set $A_{row}$ to be the set of auxiliary vertices in $G_{row}$.

First, we show that the row gadget satisfies the necessary completeness condition:
\begin{proposition}\label{prop:rowcompleteness}
For each $a \in H$, $G_{row}$ has an automorphism $f_a$ such that $f_a(b) = a + b$ for each $b \in H$.
\end{proposition}
\begin{proof}
For $j \in [t]$, define $f_j(b) = b+e_j$, where $e_j = (0, \ldots, 0, 1, 0, \ldots, 0)$, with the~$1$ in the $j$th coordinate.  If $R$ is an $i$-row in $V_x$, then $f(R)$ is also an $i$-row, for each $i$.  As a result, if we define $f$ so that it maps the star graph connected to $R$ to the star graph connected to $f_j(R)$, then $f_j$ is an automorphism on $G_{row}$.  This follows because each $i$-row is connected to a star graph with exactly $i$ outer nodes.

For more general $a \in H$, $f_a$ can be constructed by composing together many different $f_j$'s.  As automorphisms are closed under composition, this gives us an automorphism $f_a$ for each $a \in H$.
\end{proof}

Unfortunately, $G_{row}$ is not yet a variable gadget.  However, it does satisfy the following key property:
\begin{proposition}\label{prop:rowgadget}
Let $f:H \cup A_{row} \rightarrow H\cup A_{row}$ be an automorphism on $G_{row}$.  Then if $R$ is an $i$-row of $V_x$, $f(R)$ is also an $i$-row of $V_x$.
\end{proposition}
\begin{proof}
If $t = 1$ then this is trivial, as the only vertices in the row graph are those in $V_x$, and $V_x$ is the only row.  Now assume that $t > 1$.  Let $u \in A_{row}$ be the center of a star graph attached to the $i$-row $R$.  By construction, $u$ is connected to exactly $i$ degree-one nodes, so $f(u)$ must have $i$ degree-one neighbors as well.  The only degree-one vertices in $G_{row}$ are the outer nodes of the star graphs, so $f(u)$ must be the center of a star graph attached to an $i$-row.  Because $f$ is an automorphism, $f(R)$ must also be adjacent to $f(u)$.  The only way this can happen is if $f(R)$ is the $i$-row adjacent to $f(u)$.  Thus, $f(R)$ is an $i$-row if $R$ is, and this happens for each row of $V_x$.
\end{proof}

Each star graph connected to an $i$-row adds $i+1 \leq 2t$ auxiliary vertices and $i + p_i \leq t+p_i$ edges.  The number of $i$-rows is $|H|/p_i$, so this gives $2t |H|$ auxiliary vertices and $t |H| + |H| \leq 2t|H|$ edges for each coordinate $i$.  In total, the $G_{row}$ gadget contains $2t^2|H|$ auxiliary vertices and $2t^2|H|$ edges.

\paragraph{The cycle gadget:}
Consider again the group $H = \Z_3 \oplus \Z_3$.  In this case, the bijection $f: H \rightarrow H$ in which $f(a, b) = (a^3 + a, b)$ is \emph{not} an automorphism on $H$.  However, it can be extended to an automorphism on $G_{row}$.  This is because $f$ respects the order of the coordinates, and this is essentially all that the row gadget checks.  The cycle gadget will enforce that each individual coordinate behaves appropriately.

Given the gadget $G_{row}$, we extend it to a variable gadget $G$ as follows: pick a coordinate $i \in [t]$, and assume without loss of generality that $i = 1$.  Let $R$ be a $1$-row in $V_x$.  If $p_1 = 2$, then add a single edge between the two vertices in $R$.  Otherwise, $p_1 > 2$, and $R$ can be written as
\begin{equation*}
x_0 =  (0, a_2, \ldots, a_t), \quad
x_1 =  (1, a_2, \ldots, a_t),\quad \ldots, \quad
x_{p_1-1} =  (p_1 -1, a_2, \ldots, a_t).
\end{equation*}
In this case, we imagine $x_0, x_1, \ldots, x_{p_1 - 1}$ as being arranged clockwise around a circle, and between the elements $x_j$ and $x_j + 1$ (where addition is modulo $p_1$), we add the vertices and edges shown in \pref{fig:gadget}.  For example, \pref{fig:cycle} shows the result when $p_1 = 4$.  Doing this for each $i$-row and adding the result on top of $G_{row}$ completes the construction of $G$.  Set $A_{cycle}$ to be the set of additional auxiliary vertices created in this step, which is disjoint from $A_{row}$.

\begin{figure}
\centering
\begin{subfigure}{.5\textwidth}
	\centering
	\includegraphics{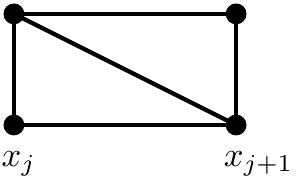}
	\caption{The basic unit of the cycle gadget.  The vertices labeled $x_j$ and $x_{j+1}$ are variable vertices, whereas the two unlabeled vertices are auxiliary vertices.}
	\label{fig:gadget}
\end{subfigure}%
\begin{subfigure}{.5\textwidth}
	\centering
	\includegraphics{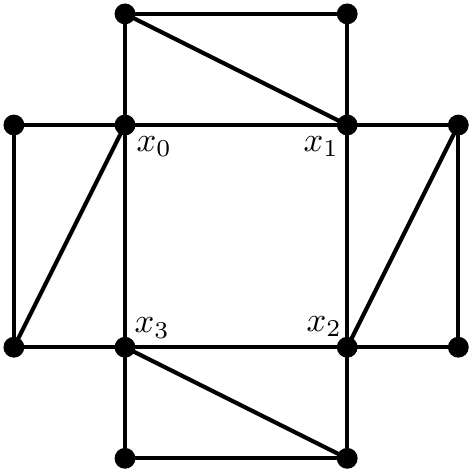}
	\caption{The cycle gadget applied to $\Z_4$.}
	\label{fig:cycle}
\end{subfigure}
\caption{The cycle gadget construction.  Each $i$-row has such a gadget placed on top of it.}
\end{figure}

We first show that $G$ satisfies the completeness condition:
\begin{proposition}
For each $a \in H$, $G$ has an automorphism $f_a$ such that $f_a(b) = a + b$ for each $b \in H$.
\end{proposition}
\begin{proof}
Define $f_j$ on the variable vertices and the star graphs as in \pref{prop:rowcompleteness}.  We need to define $f_j$ on the cycle gadgets as well, and the same reasoning as before shows that because $f_j$ maps $i$-rows to $i$-rows (and it maintains the relative order of the vertices in each rows), we can set $f_j$ to locally map the cycle gadget from row $R$ to row $f_j(R)$ so that $f_j$ is an automorphism on $G$.  And as before, it is easy to extend this to hold for all $f_a$.
\end{proof}
In addition, it satisfies the soundness condition:
\begin{proposition}
Let $f$ be an automorphism on $G$.  Then there exists an element $a \in H$ such that $f(b) = a+ b$ for each $b \in H$.
\end{proposition}
\begin{proof}
Let $f$ be an automorphism on $G$, and let $R$ be an $i$-row of $V_x$.  Note that none of the vertices in $A_{cycle}$ have degree one.  Therefore, the only degree-one vertices in $G$ are in the star graphs from the row gadget.  This is enough to reuse the argument in \pref{prop:rowgadget} to show that $f(V_x) = V_x$, $f(A_{row}) = A_{row}$, and $f(A_{cycle}) = A_{cycle}$.  In addition, we may apply \pref{prop:rowgadget} to $R$, which shows that $f(R)$ is also an $i$-row of $V_x$.

Write the elements of $R$, as in the above paragraph, as $x_0, \ldots, x_{p_i - 1}$, where $x_j$ has a $j$ in the $i$th coordinate.  We have the relationship $x_{j+1} = x_j + e_i$.  Let $u$ be the vertex adjacent to both $x_j$ and $x_{j+1}$ in \pref{fig:gadget}, and let $v$ be the vertex adjacent to both $u$ and $x_{j+1}$.  The only row that $x_j$ and $x_{j+1}$ are both in is $R$.  This means that $u$ is the \emph{only} vertex in $A_{cycle}$ adjacent to both $x_j$ and $x_{j+1}$, and likewise $v$ is the only vertex in $A_{cycle}$ adjacent to both $u$ and $x_{j+1}$.  This means that if $f$ is an automorphism, $f(u)$ must be the only vertex in $A_{cycle}$ adjacent to both $f(x_j)$ and $f(x_{j+1})$, and $f(v)$ must be the only vertex in $A_{cycle}$ adjacent to both $f(u)$ and $f(x_{j+1})$.  The only way this is possible is if $f(x_{j+1}) = f(x_j) + e_i$.  Rewriting $x_{j+1}$ as $x_j + e_i$, we conclude that $f(x_j + e_i) = f(x_j) + e_i$.

The above argument only used that $x_j$ and $x_{j+1}$ differed by one in the $i$th coordinate.  As a result, we know that $f(b + e_i) = f(b) + e_i$ for any coordinate $i$.  More generally, we know that $f(a+b) = f(a) + b$, which can be shown by applying the previous fact coordinate-by-coordinate.  Set $a:= f(0)$.  Then $f(b) = f(0) + b = a+b$.  This completes the soundness case.
\end{proof}

Each time the cyclic gadget is applied to an $i$-row, it adds $2 p_i$ auxiliary vertices and $5 p_i$ edges to the graph.  As the number of $i$-rows is $|H|/p_i$, this adds $2 |H|$ vertices and $5 |H|$ edges per coordinate $i$, for a total of $2t|H|$ vertices and $5t|H|$ edges.  Combining this with the row gadget, the vertex gadget creates at most $4t^2 |H|$ vertices and $7t^2 |H|$ edges.  Using the fact that $p_1\cdots p_t = |H|$, we see that $t \leq \log_2 |H|$, giving the bound in the lemma statement.

As for the size of the maximum clique, we consider three cases. If $|H| =2$, then the variable gadget is a $2$-clique.  If $|H| = 3$, then the variable gadget is as in \pref{fig:cycle}, except with a triangle in the center rather than a square.  It can be checked that the largest clique in this case is of size three.  Finally, if $|H| > 3$, any clique containing an auxiliary variable from a cycle gadget must be of size at most three, as \pref{fig:gadget} verifies.  Next, any group of variable vertices can appear in a clique together only if they appear in the same $i$-row together.  Furthermore, the largest clique in an $i$-row is of size two unless $p_i = 3$, meaning the three variable vertices in that $i$-row form a triangle.  In this case (supposing $t >1$), they will also all be connected to the center of a star graph by the row gadget.  This is the only way to form a clique of size four, and no larger clique is possible.  Thus, the maximum clique size is at most $\min\{4, |H|\}$.
\end{proof}

\subsection{Label-extended graph proof}\label{sec:labelext}

\begin{proof}[Proof of \pref{lem:constraintvariable}]
We handle the completeness and soundness cases separately:
\paragraph{Completeness.}
Given $\alpha$, define $f$ as follows: first, $f(x_i \mapsto b) := x_i \mapsto (b-\alpha(x_i))$ for each $i \in [k]$ and $b \in H$.  Next, for each assignment $\beta$ which satisfies $C$, set $f(\beta) := \beta - \alpha$, using coordinate-wise subtraction.  For this definition to be legitimate, we must ensure that $\beta-\alpha$ is a satisfying assignment of $\pred$.  Because $C$ is an $\ABEL$ constraint, it can be written as
\begin{equation*}
\pred(x_1 + a_1, \ldots, x_k + a_k) = 1.
\end{equation*}
Because $\alpha$ and $\beta$ both satisfy $C$, they can be written as $\alpha' - (a_1, \ldots, a_k)$ and $\beta' - (a_1, \ldots, a_k)$, where $\alpha'$ and $\beta'$ are satisfying assignments of $\phi$.  Then by the subgroup property of $\phi$, $\beta-\alpha = \beta'-\alpha' \in \phi$.  Therefore, $f$ is a legitimate bijection between the vertices of $G$ and $\homog{G}$.

Given an assignment $\beta$, there is an edge in $G$ between $\beta$ and $x_i \mapsto \beta(i)$, for each $i \in [k]$.  Then $f(\beta) = \beta- \alpha$, and $f(x_i \mapsto \beta(i)) =x_i \mapsto (\beta(i)- \alpha(i))$, both of which are connected by an edge in $\homog{G}$.  Thus, $f$ is a homomorphism.

\paragraph{Soundness.}
Let $\beta$ be the assignment in $V_{C}$ which $f$ maps to the all-$0$'s assignment.  Then for each $i\in[k]$, $x_i \mapsto \beta(i)$ is connected to $\beta$ in $G$, so $f(x_i\mapsto \beta(i)) = x_i \mapsto (\beta(i) - \alpha(i))$ must be connected to $0$ in $\homog{G}$.  This means that $\beta(i) - \alpha(i) = 0$ for all $i \in[k]$.  In particular, $\alpha = \beta$, and because $\beta$ is a vertex in $V_C$, it (and by extension $\alpha$) is a satisfying assignment of $C$.
\end{proof} 
\fi

\section{Conclusions} \label{sec:conclusions}

We have shown SOS/Lasserre gaps for \giso and hardness of \robustgiso based on various average case hardness assumptions. The gaps we obtained are tiny. Although we did not attempt to optimize over the parameters in order to maximize the gaps, we believe the current approach cannot make the soundness as small as an arbitrary constant. Therefore one intriguing problem at this point is to find a way amplifying the gaps. A possible approach is to prove a ``parallel repetition theorem for \giso'', i.e. find a graph product operator $\square$ such that $\GI(G^{\square \ell}, H^{\square \ell}) \leq \GI(G, H)^{\omega_{\ell}(1)}$, where $\omega_{\ell}(1)$ means a function which goes to $+\infty$ as $\ell \to +\infty$. A natural candidate for $\square$ is the tensor product $\tensor$. It is known that (\cite{McK71}), for (exact) \gisolong, if both $G$ and $H$ are connected and non-bipartite graphs, $G \tensor G$ is not isomorphic to $H \tensor H$ when $G$ is not isomorphic to $H$. But how the tensor product works with approximate \gisolong remains to be explored. 

Another future direction is to prove robust asymmetry property for random $d$-regular graphs. Since in this model we do not have the problem of many isolated vertices, even for constant $d$, we might be able to prove a $(\beta, \Omega(\beta))$-asymmetry property for beta being as small as $\Theta\left(\frac{1}{n}\right)$, where in contrast, there is an $\exp(-O(c))$ lower bound for $\beta$ in the $\calG_{n, cn}$ model.



\bibliographystyle{alpha}
\bibliography{GISO-Lasserre}

\end{document}